\documentclass{article}
\usepackage{custom_tex}
\usepackage{authblk}
\title{Markov Chain Monte Carlo with Diffusion Paths}
\date{September 2026}
\author[1]{Han Chen}
\author[1]{Sifan Liu}
\author[2]{Jun Yang}
\affil[1]{Department of Statistical Science, Duke University}
\affil[2]{Department of Mathematical Sciences, University of Copenhagen}

\begin{document}
\maketitle

\begin{abstract}
Sampling from multimodal distributions is a longstanding challenge for classical local Markov chain Monte Carlo (MCMC) methods. A popular remedy is to introduce a sequence of intermediate distributions that interpolate between the target and a simpler reference. The classical choice, tempering, raises the density to a power, but distorts the relative weights of asymmetric modes and can lead to poor mixing. We instead propose interpolating along the diffusion path, the marginals of a noising diffusion process that carries the target toward a Gaussian. This path preserves the relative weights of the modes and enjoys favorable mixing properties, which we make precise through a spectral-gap analysis of the corresponding ideal transition kernel. Sampling along the path requires its intermediate scores, which can be estimated from the unnormalized target through variational approaches, yielding only an approximate sampler. To remove the resulting bias, we introduce the Metropolis-adjusted diffusion path (MAD-Path) sampler, which corrects the diffusion-path proposal in an augmented path space and leaves the target invariant regardless of the accuracy of the learned score or the discretization error. We further quantify how these two errors affect the acceptance probability, providing guidance for practical tuning. Experiments on a range of Bayesian posteriors show that MAD-Path improves global exploration and mode-weight estimation relative to tempering-based MCMC methods and unadjusted diffusion samplers.
\end{abstract}

\tableofcontents

\section{Introduction}\label{sec:introduction}

Markov chain Monte Carlo (MCMC) is the standard framework for sampling from a density known only up to its normalizing constant, and it serves as the computational engine of Bayesian statistics and many scientific applications. The most widely used algorithms---random-walk Metropolis \citep[RWM;][]{metropolis1953equation}, Metropolis-adjusted Langevin algorithm \citep[MALA;][]{roberts1996exponential}, and Hamiltonian Monte Carlo \citep[HMC;][]{duane1987hybrid, neal2011mcmc}---generate proposals through local moves, and consequently mix poorly when the target has several well-separated modes, as the chain only rarely crosses the low-probability barriers between them. Such multimodal targets arise routinely in Bayesian posteriors, through label switching in mixture models \citep{Jasra2005Markov}, parameter constraints \citep{Sundberg2010FlatAM}, and nonidentifiability \citep{drton2004multimodality}, as well as throughout the sciences; see \citet{latuszynski2025mcmc} for a recent survey.

A general strategy for sampling from multimodal distributions is to introduce a path of intermediate distributions that connects the target $p$ to a flatter distribution on which the separated modes are easier to traverse. The classical construction is tempering: for a decreasing sequence of inverse temperatures $1=\beta_0>\beta_1>\cdots>\beta_n>0$, one considers intermediate densities proportional to $p^{\beta_k}$. As $\beta_k$ decreases, differences in log-density are compressed, reducing the energy barriers between modes and allowing local Markov transitions to move more freely across regions that are nearly disconnected under the original target. This idea underlies simulated tempering \citep{marinari1992simulated,Geyer1995AnnealingMC}, parallel tempering \citep{swendsen1986replica,geyer1991markov} and its non-reversible variants~\citep{okabe2001replica,syed2022non}, and tempered transitions \citep{neal1996sampling}, among other methods. Closely related annealing paths also form the basis of annealed importance sampling \citep{neal2001annealed} and sequential Monte Carlo samplers \citep{del2006sequential}.

Tempering can, however, fail when the local modes exhibit asymmetry. For instance, when the modes differ in scale, raising the density to a power distorts their relative weights.  Consider an equal-weight mixture of two well-separated Gaussians with different scales. Under the tempered density $p^{\beta}$ with $\beta<1$, the narrower mode has a weight that is exponentially small in the dimension. This weight distortion is known to cause torpid mixing of simulated and parallel tempering for asymmetric mixtures, with spectral gaps that vanish exponentially in dimension regardless of the temperature schedule~\citep{woodard2009torpid}. A related phenomenon has been described in the learning-based sampling literature as ``teleportation of mass'': when the tempering interpolation is viewed as a curve in Wasserstein probability space, probability mass is transported nonlocally along the path, leading to an irregular velocity field~\citep{mate2023learning,chemseddine2025neural}. Weight-preserving variants of tempering \citep{Tawn2018WeightpreservingST,tawn2020annealed} target this issue directly, but they require knowing the location and scale of every mode and therefore do not resolve it in general.

This shortcoming of tempering motivates the search for a better interpolating path. Crucially, most of the algorithms above---simulated and parallel tempering, tempered transitions, and their relatives---do not depend on the tempering path in particular; they require only a sequence of intermediate distributions whose unnormalized densities can be evaluated, and any such sequence may be substituted. A particularly attractive alternative is to interpolate at the level of the samples rather than by powering the density. Let $X\sim p$ and $Z\sim\N(0,I_d)$ be independent, and consider $X_t = e^{-t/2}X + \sqrt{1-e^{-t}}Z$ for $t\geq 0$. As $t$ runs from $0$ to $\infty$, the law of $X_t$ moves from the target $p$ toward the standard Gaussian $\N(0,I_d)$. Unlike tempering, this path preserves mixture weights exactly: if the target is a mixture distribution, the law of $X_t$ is again a mixture with the same weights, whatever the relative scales of the components.

The marginal laws are precisely the time marginals of the Ornstein--Uhlenbeck (OU) diffusion
\begin{align*}
    \rd X_t = -\tfrac12 X_t \,\rd t + \rd B_t,\qquad X_0\sim p,\;\; t\geq 0,
\end{align*}
which is the overdamped Langevin dynamics converging exponentially fast to its stationary distribution $\N(0, I_d)$. Writing $p_t$ for the density of $X_t$, the family $(p_t)_{t\geq 0}$ solves the Fokker--Planck equation
\begin{align*}
    \partial_t p_t + \nabla\cdot(p_t\, v_t) = 0,
    \qquad
    v_t(\bfx) = -\tfrac12\bigl(\bfx + \nabla\log p_t(\bfx)\bigr).
\end{align*}
Hence, $(p_t)_{t\geq 0}$ is an absolutely continuous curve in the Wasserstein space $(\mathcal P_2(\R^d), W_2)$ driven by the velocity field $v_t$~\citep{ambrosio2008gradient}. In contrast to tempering, mass is then transported along $(p_t)_{t\geq 0}$ by a smooth velocity field rather than teleported non-locally, making the diffusion path a more regular curve in Wasserstein space.

The difficulty in sampling with the diffusion path is that the intermediate distributions are intractable. To transport a Gaussian sample to the target $p$, one must run the noising diffusion backward in time, with a drift that involves the score $\nabla\log p_t$ at every intermediate time~\citep{anderson1982reverse}. This reverse-time construction is the basis of score-based generative modeling~\citep{song2019generative, song2021scorebased, ho2020denoising}, where the scores are estimated from data by denoising score matching. 
In generative modeling, diffusion models have become a leading class of flexible samplers. Their success is due not only to the expressiveness of neural networks for estimating the score, but also to the structure of the diffusion path itself: the forward process gradually transforms a complex data distribution into noise, so that generation can be formulated as a sequence of simpler denoising problems rather than a single global transport problem. This is closely analogous to the role of annealing paths in tempering-based MCMC~\citep{song2019generative}.

In the sampling setting, however, one has no samples from $p$ and only access to its unnormalized density. Recent work therefore estimates the intermediate-time scores through stochastic optimal control or path-space variational formulations~\citep{zhang2022path, richter2024improved, vargasdenoising}. The resulting diffusion sampler is only approximate: its output distribution is affected by score-estimation error, finite-time truncation of the diffusion path, and numerical discretization~\citep{Chen2022SamplingIA}. Moreover, recent empirical studies suggest that successful applications to multimodal targets often require incorporating the target score $\nabla\log p$ directly into the score parametrization; without this information, the learned sampler can suffer from severe mode collapse~\citep{blessing2024beyond,he2025no}. This is costly, however, since the target score must then be evaluated at every discretization step. Mode collapse can be alleviated further by encoding the target's mode structure into the reference distribution, for instance, by locating the modes from preliminary samples~\citep{noble2025learned}. But even then, the sampler can remain inaccurate when the target has a strongly nonlinear score that is hard to approximate. We return to these practical aspects of score learning in Section~\ref{sec:practical}.

Unlike other approximate samplers, such as normalizing flows, diffusion-based samplers do not provide a simple bias-correction mechanism. A normalizing flow is a deterministic invertible map from a tractable reference distribution, so the change-of-variables formula gives the exact density of each generated sample. This density can then be used to debias the sampler by importance weighting or a Metropolis--Hastings correction.
This convenience, however, comes at a price. First, requiring an invertible map with a tractable Jacobian determinant greatly constrains the expressiveness of normalizing flows~\citep{papamakarios2021normalizing}. Second, in dimension two or higher, infinitely many diffeomorphisms carry one density to the other, and it is unclear which one a normalizing flow should learn. Diffusion-based samplers, in contrast, place no invertibility constraint on the parametrization, and because they prescribe the entire marginal path, the object to be learned is uniquely determined.
The tradeoff is that samples are generated through a stochastic path, and the marginal density of the endpoint is generally intractable. Although it can be evaluated through the associated probability-flow ODE~\citep{song2021scorebased}, this requires solving an additional ODE and, in practice, introduces numerical discretization error. Thus, diffusion-based samplers do not come with exact density evaluation, making their sampling bias difficult to remove by standard correction methods.

To address this difficulty, we propose an MCMC algorithm that uses diffusion paths as proposals and corrects the resulting bias with a Metropolis--Hastings step. In each transition, the current state is first evolved forward along the noising diffusion path, then back along an approximate reverse path. This forward-then-backward construction mirrors the heating-then-cooling structure of Neal's tempered transitions~\citep{neal1996sampling}, with the diffusion path in place of the tempering path. Near the noisy end, the target's modes are smoothed together, so the chain can move between them freely; the reverse path then transports the state back toward $p$. Because the proposal is a stochastic trajectory, the Metropolis--Hastings correction is applied to the entire path in an augmented space. This augmented-space correction ensures that the marginal chain has $p$ as its exact stationary distribution, regardless of the accuracy of the learned score or the size of the discretization error.

The main contributions and the outline of the paper are as follows.
\begin{itemize}
\item In Section~\ref{sec:ideal:transition} we describe the ideal transition kernel induced by the diffusion path and study its mixing, first when the target satisfies a Poincar\'e inequality and then when it is a mixture of such distributions. This ideal kernel coincides, up to a rescaling, with the proximal sampler~\citep{lee2021structured, chen2022improved}. Our analysis takes a new perspective, interpreting the proximal sampler as a data-augmentation Gibbs sampler. This lets us bound its spectral gap without appealing to isoperimetric inequalities and, for mixture targets, decompose the gap into a between-mode and a within-mode contribution, making precise the advantage of the diffusion path as an interpolation.

\item In practice, the diffusion-path transition can only be run with approximate scores and a finite-time discretization. In Section~\ref{sec:algorithm}, we introduce the Metropolis--Hastings algorithm that removes the resulting score-approximation and discretization bias. 
Because the proposal is constructed by a stochastic trajectory, the correction is carried out in an augmented path space. This contrasts with existing Metropolis adjustments for deterministic normalizing flows, where the flow is used either to generate independent proposals~\citep{gabrie2022adaptive} or to reparametrize the target before applying local MCMC~\citep{hoffman2019neutra}. 

\item Section~\ref{sec:analysis} quantifies how the Metropolis acceptance probability depends on discretization and score approximation. With the exact score, we isolate the discretization error and show that the rejection probability is of order $O(h^{1/2})$ for step size $h$. We then derive the corresponding high-dimensional scaling rule and the associated optimal tuning guideline. With an approximate score, we show that the rejection probability is governed by the symmetrized Kullback--Leibler divergence between the approximate and exact reverse diffusion, or equivalently, an integrated squared score error.

\item In Section~\ref{sec:practical}, we describe how the intermediate-time scores can be learned through path-space variational inference and discuss related work along this line. We also give practical guidelines for choosing the reference distribution, which we find crucial for mitigating mode collapse.

\item Finally, in Section~\ref{sec:experiments}, we evaluate the method on a range of Bayesian posteriors, study how diffusion time and step size affect sampling efficiency, and compare against tempering-based MCMC. Moreover, we show that scores learned by existing variational approaches can misrepresent the relative weights of isolated modes, or be inaccurate when the target's score is highly nonlinear. In both cases, the Metropolis--Hastings correction is crucial for removing the resulting bias.
Section~\ref{sec:discussion} concludes with a discussion of future directions.
\end{itemize}

\section{Diffusion paths as global transitions}
\label{sec:ideal:transition}

Bridging the target with a tractable reference through a sequence of intermediate distributions is a common strategy for constructing MCMC transitions that can move globally across separated modes. Here, we take the diffusion path as this bridge and argue that it is preferable to classical tempering.

\subsection{Forward--backward diffusion path as an ideal transition kernel}
Let $p$ be the target distribution supported on $\R^d$ with continuously differentiable density. Recall the Ornstein--Uhlenbeck (OU) process
\begin{align}\label{equ:ou}
    \rd X_t = -\frac12 X_t \rd t + \rd B_t, \quad X_0\sim p,\;\; t\in[0,T],
\end{align} 
where $B_t$ is a standard $d$-dimensional Brownian motion, now run over a finite horizon $[0,T]$. Writing $p_t$ for the law of $X_t$, the family $(p_t)_{t\in[0,T]}$ interpolates between the target $p_0=p$ and $p_T$; for large $T$ the endpoint $p_T$ is close to the equilibrium distribution $\N(0,I_d)$, but need not equal it. The transition kernel constructed below is reversible with respect to $p$ for every finite $T$, and $T$ governs only the mixing speed.

The time-reversal of the OU process~\eqref{equ:ou} is given by
\begin{align}\label{equ:ou:backward}
    \rd Y_t = \Big(-\frac12 Y_t - \nabla \log p_t(Y_t) \Big) \widebar{\rd t} + \rd \widebar B_t, \quad Y_T\sim p_T,\, t\in[T,0],
\end{align}
where $\widebar{\rd t}$ denotes an infinitesimal negative time increment and $\widebar B_t$ is a standard Brownian motion running backward in time~\citep{anderson1982reverse}. By construction, the forward path $(X_t)_{t\in[0,T]}$ and the backward path $(Y_t)_{t\in[0,T]}$ are equal in distribution. This suggests a natural Markov transition kernel: from the current state $X_0$, sample the forward path and set $Y_T=X_T$; then sample the backward path from $Y_T$ and return its endpoint $Y_0$ as the next state.

\begin{prop}[Reversibility of the diffusion-path transition kernel]
\label{prop:ideal:kernel:reversible}
Let $P(\bfx_0, \rd \bfy_0)$ be the transition kernel that draws $\bfy_0$ by first sampling the forward path $(X_t)_{t\in[0,T]}$ from the OU process~\eqref{equ:ou} started at $X_0=\bfx_0$, and then the backward path $(Y_s)_{s\in[T,0]}$ from~\eqref{equ:ou:backward} started at $Y_T=X_T$. Then $P$ is reversible with respect to $p$, and hence admits $p$ as its invariant distribution.
\end{prop}
\begin{proof}[Proof of Proposition~\ref{prop:ideal:kernel:reversible}]
Write $\vec{P}(\bfx_0, \rd \bfx_T)$ for the forward transition from $\bfx_0$ to $\bfx_T$ induced by~\eqref{equ:ou}, and $\leftvec{P}(\bfx_T, \rd \bfy_0)$ for the backward transition from $\bfx_T$ to $\bfy_0$ induced by~\eqref{equ:ou:backward}, so that
\begin{align*}
    P(\bfx_0, \rd \bfy_0) = \int \vec{P}(\bfx_0, \rd \bfx_T)\, \leftvec{P}(\bfx_T, \rd \bfy_0).
\end{align*}
By construction, the forward and backward processes satisfy the reversal identity $p_0(\bfx_0)\, \vec{P}(\bfx_0, \rd \bfx_T) = p_T(\bfx_T)\, \leftvec{P}(\bfx_T, \rd \bfx_0)$. Hence
\begin{align*}
    p_0(\bfx_0)\, P(\bfx_0, \rd \bfy_0) &= \int p_0(\bfx_0)\, \vec{P}(\bfx_0, \rd \bfx_T)\, \leftvec{P}(\bfx_T, \rd \bfy_0)\\
    &= \int p_T(\bfx_T)\, \leftvec{P}(\bfx_T, \rd \bfx_0)\, \leftvec{P}(\bfx_T, \rd \bfy_0) \\
    &= \int p_0(\bfy_0)\, \vec{P}(\bfy_0, \rd \bfx_T)\, \leftvec{P}(\bfx_T, \rd \bfx_0)  \\
    &= p_0(\bfy_0)\, P(\bfy_0, \rd \bfx_0).
\end{align*}
This is detailed balance, so $P$ is reversible with respect to $p_0=p$.
\end{proof}

The argument does not rely on the specific form of the forward and backward processes~\eqref{equ:ou}--\eqref{equ:ou:backward}. The key property is that the forward and backward processes are constructed to be time-reversals of each other, which ensures the detailed balance condition.

However, the ideal kernel $P$ is not directly implementable, as it requires simulating the continuous-time diffusions and evaluating the score $\nabla \log p_t$ at every intermediate time. In Section~\ref{sec:algorithm}, we introduce a Metropolis--Hastings algorithm that remains reversible with respect to the target even when the diffusion paths are only approximately simulated. Before that, we analyze the spectral gap of the ideal kernel to quantify the benefit of interpolating along the diffusion path.

\subsection{Spectral gap of the diffusion-path transition kernel}

We start from a classical characterization of the spectral gap of $P$ in terms of maximal correlation. The transition kernel $P$ is the marginal chain on $X_0$ of a two-block data-augmentation Gibbs sampler on $(X_0,X_T)$, where $X_T\mid X_0\sim\vec P$ and $X_0\mid X_T\sim\leftvec P$. For such a sampler, the spectral gap can be read off from the maximal correlation between $X_0$ and $X_T$, taken under their joint law $p_0(\rd\bfx_0)\,\vec P(\bfx_0,\rd\bfx_T)$, defined by 
\begin{align*}
    \rho_{\max}(X_0,X_T)^2 = \sup_{\substack{f\in L^2(p_0),\ g\in L^2(p_T)\\ f,\,g\ \text{non-constant}}} \mathrm{Corr}\big(f(X_0),\,g(X_T)\big)^2.
\end{align*}
Writing $(Pf)(\bfx_0) := \int f(\bfy_0)\, P(\bfx_0, \rd \bfy_0)$, the spectral gap of $P$ is
\begin{align*}
    \mathrm{Gap}(P)=1-\sup_{\substack{f\in L^2(p_0)\\ \EE[p_0]{f}=0,\ f\neq 0} } \frac{\langle f, P f\rangle_{L^2(p_0)}}{\|f\|_{L^2(p_0)}^2}.
\end{align*}

\begin{lem}[Spectral gap via maximal correlation]\label{lem:gap:maximal:corr}
The spectral gap of the transition kernel $P$ satisfies
\begin{align*}
    \mathrm{Gap}(P)=1-\rho_{\max}(X_0,X_T)^2.
\end{align*}
\end{lem}

The proof is given in Appendix~\ref{prf:gap:maximal:corr}.
The link between the spectral gap of a two-block Gibbs sampler and maximal correlation is classical in the analysis of data-augmentation algorithms. \citet{liu1994covariance} use it to study the covariance structure of Gibbs samplers and to compare the efficiency of different data-augmentation schemes, and the same quantity appears in the mixing theory of stationary processes and Markov chains as the $\rho$-mixing coefficient~\citep{bradley2005basic}. In the remainder of this section, we use Lemma~\ref{lem:gap:maximal:corr} to bound $\mathrm{Gap}(P)$ under increasingly general assumptions on the target: first when $p$ satisfies a Poincar\'e inequality, and then when $p$ is a mixture of such distributions.

\subsubsection{Spectral gap under a Poincar\'e inequality}

Because the forward process $(X_t)_{t\in[0,T]}$ solves the OU SDE~\eqref{equ:ou}, the forward kernel $\vec{P}$ is Gaussian,
\begin{align*}
X_T\mid X_0\sim \N\big(\sqrt{\alpha}\, X_0,\, (1-\alpha)I_d\big),\qquad \alpha:=e^{-T}.
\end{align*}
Thus $X_T$ is a contraction of $X_0$ toward the origin, perturbed by independent Gaussian noise. The parameter $\alpha\in(0,1)$ decreases as the horizon $T$ grows: it is near $1$ for small $T$, where $X_T\approx X_0$, and near $0$ for large $T$, where $X_T$ is almost independent of $X_0$. In this way, $\alpha$ controls the correlation between $X_0$ and $X_T$. When $p$ satisfies a Poincar\'e inequality, this correlation can be bounded from above, and Lemma~\ref{lem:gap:maximal:corr} converts the bound into a lower bound on the spectral gap.

\begin{thm}[Spectral gap under a Poincar\'e inequality]
\label{thm:gap:poincare}
Suppose $p$ satisfies the Poincar\'e inequality
\begin{align}\label{equ:poincare}
    \Var[p]{f} \leq C_{\mathrm P} \int \|\nabla f(\bfx)\|^2 \, p(\rd \bfx)
    \qquad\text{for all smooth } f.
\end{align}
Then
\begin{align*}
    \rho_{\max}(X_0,X_T)^2
    \le
    \frac{\alpha C_{\mathrm P}}
    {1-\alpha+\alpha C_{\mathrm P}},
    \qquad\text{so that}\qquad
    \mathrm{Gap}(P) \geq \frac{1-\alpha}{1-\alpha+\alpha C_{\mathrm {P}}}.
\end{align*}
Both bounds are tight for Gaussian targets: if $p=\N(\mu,\Sigma)$, its Poincar\'e constant is $C_{\mathrm P}=\lambda_{\max}(\Sigma)$, the largest eigenvalue of $\Sigma$, and 
\begin{align*}
    \rho_{\max}(X_0,X_T)^2 = \frac{\alpha \lambda_{\max}(\Sigma)}{1-\alpha + \alpha \lambda_{\max}(\Sigma)},
    \qquad
    \mathrm{Gap}(P)=\frac{1-\alpha}{1-\alpha+\alpha\lambda_{\max}(\Sigma)}.
\end{align*}
\end{thm}

The proof is given in Appendix~\ref{prf:gap:poincare}. 
The result makes precise the role of the diffusion time $T$, or equivalently $\alpha=e^{-T}$, in the convergence of the ideal kernel $P$. As $T\to\infty$ ($\alpha\to0$), the correlation between $X_0$ and $X_T$ vanishes and the spectral gap tends to $1$. Conversely, as $T\to0$ ($\alpha\to1$), the maximal correlation tends to $1$ and the gap tends to $0$, since the forward step barely perturbs the current state.
Quantitatively, the lower bound on $\mathrm{Gap}(P)$ implies that the $\chi^2$-divergence contracts by a factor of $(\frac{\alpha C_{\mathrm P} }{1-\alpha + \alpha C_{\mathrm P} })^2$ per transition, which provides a mixing time bound for the ideal transition kernel $P$.

\begin{rmk}[Connection to the proximal sampler]
The ideal transition kernel $P$ coincides, up to rescaling, with the proximal sampler~\citep{lee2021structured, chen2022improved, wibisono2026mixing}. 
Under the Poincar\'e inequality, \citet[Theorem 4]{chen2022improved} showed that the $\chi^2$-divergence contracts by a factor of $\frac{1}{(1+\eta/ C_{\mathrm P})^2}$ per iteration, where $\eta=\frac{1-\alpha}{\alpha}$ is the step size. 
This is exactly the rate obtained from our spectral-gap bound in Theorem~\ref{thm:gap:poincare}.
The analysis in~\cite{chen2022improved}, however, is based on the dissipation of $\chi^2$-divergence along the forward and backward heat flows. 
In contrast, our argument uses the structural relation between the spectral gap and the maximal correlation.
This viewpoint allows us to further analyze the spectral gap without global functional inequalities, as discussed below.
\end{rmk}

\subsubsection{Spectral gap for a mixture of Poincar\'e components}

When the target is multimodal, a global Poincar\'e inequality may fail. We therefore assume a mixture structure
\[
    p=\sum_{c\in\calC} w_c\, p_c,
    \qquad w_c>0,\quad \sum_{c\in\calC} w_c=1,
\]
with each component $p_c$ a probability distribution on $\R^d$, and analyze the gap of $P$ by lifting to the augmented space $\R^d\times\calC$.
The lifted target is $\tp(\bfx,c)=w_c p_c(\bfx)$, and the lifted kernel $\widetilde P$ acts as $(X_0,C)\to X_T\to(X_0',C')$: given $(X_0,C)$ we draw $X_T\sim\N(\sqrt\alpha X_0,(1-\alpha) I_d)$, then sample $(C',X_0')$ from the posterior of $(C,X_0)\mid X_T$. The $X$-marginal of the lifted transition $\widetilde P$ is exactly $P$, 
and since $X_T$ is conditionally independent of $C$ given $X_0$,
\[
\tilde\rho_{\max}^2:=\rho_{\max}\big((X_0,C),X_T\big)^2 = \rho_{\max}(X_0,X_T)^2.
\]

The next result bounds $\tilde\rho_{\max}^2$ in terms of the squared label correlation $\rho_{\mathrm{lab}}^2:=\rho_{\max}(C,X_T)^2$ and the squared within-component correlations $\rho_c^2:=\rho_{\max}(X_0,X_T\mid C=c)^2$, the latter computed under the conditional law of $(X_0,X_T)$ given $C=c$.

\begin{prop}[Maximal-correlation bound for mixtures]
\label{prop:lifted:corr:decomposition}
With the notation above,
\[
    \tilde\rho_{\max}^2
    \leq
    \rho_{\mathrm{lab}}^2
    +
    \bigl(1-\rho_{\mathrm{lab}}^2\bigr)
    \max_{c\in\calC}\rho_c^2.
\]
Consequently,
\[
    \mathrm{Gap}(P)
    \geq
    \bigl(1-\rho_{\mathrm{lab}}^2\bigr)
    \left(1-\max_{c\in\calC}\rho_c^2\right).
\]
\end{prop}
The proof is given in Appendix~\ref{prf:lifted:corr:decomposition}, which uses the law of total variance to decompose the correlation into between- and within-component contributions, a classic technique in Markov-chain analysis~\citep{madras2002markov,ge2018simulated}.

Let $P_{\mathrm{lab}}$ denote the projection of $\widetilde P$ onto the label space:
$P_{\mathrm{lab}}(c,c')= \EE{\bbP(C'=c'\mid X_T) \mid C=c }$.
Let $P_c$ denote the transition kernel $X_0\to X_T\to X_0'$ under the conditional law given $C=c$.
Then Proposition~\ref{prop:lifted:corr:decomposition} can be equivalently stated as
\begin{align*}
\mathrm{Gap}(P) \geq \mathrm{Gap}(P_{\mathrm{lab}}) \cdot \min_{c\in \calC} \mathrm{Gap}(P_c).
\end{align*}
When each component $p_c$ satisfies a Poincar\'e inequality, Theorem~\ref{thm:gap:poincare} provides a lower bound on $\min_{c\in\calC}\mathrm{Gap}(P_c)$. 
Once the within-component mixing is controlled, the resulting bound is governed by $\mathrm{Gap}(P_{\mathrm{lab}})$, which measures the chain's ability to move between components. Intuitively, transitions become more difficult as the components grow further apart. The following theorem quantifies this intuition by lower-bounding $\mathrm{Gap}(P_{\mathrm{lab}})$ in terms of the Wasserstein diameter of the components.

\begin{thm}[Spectral gap for a mixture of Poincar\'e distributions]
\label{thm:mixture:poincare:gap:bound}
Let the target be $p=\sum_{c\in\calC} w_c p_c$ with $|\calC|<\infty$,
and suppose each $p_c$ satisfies the Poincar\'e
inequality~\eqref{equ:poincare} with a constant $C_{\mathrm P}$.
Then
\[
    \mathrm{Gap}(P)
    \geq
    \exp\left\{
        -\frac{\alpha}{4(1-\alpha)}
        \max_{c,c'\in\calC} W_2^2(p_c,p_{c'})
    \right\} 
    \frac{1-\alpha}{1-\alpha+\alpha C_{\mathrm P}}.
\]
In particular, the lower bound is uniform over the mixture weights $(w_c)_{c\in\calC}$.
\end{thm}

The proof is given in Appendix~\ref{prf:mixture:poincare:gap:bound}. The bound separates two factors governing the gap of the ideal kernel $P$: the between-component separation, quantified by the Wasserstein diameter $\max_{c,c'\in\calC}W_2(p_c,p_{c'})$, and the within-component conditioning, quantified by the Poincar\'e constant.
Notably, the bound is uniform over the mixture weights and contains no explicit dependence on the ratios of the component scales. This contrasts with tempering, where weight or scale imbalances can distort the relative mode weights along the tempering path and lead to poor mixing.
Moreover, the use of the Wasserstein diameter can be conservative, as it replaces all pairwise overlaps by the smallest one. A sharper analysis could retain the individual pairwise overlaps and bound the label-chain gap through the corresponding weighted graph on $\calC$.

\paragraph{Mixture of Gaussians with unequal variances}

We now instantiate the results on a mixture of two Gaussians with unequal variances, the target \citet{woodard2009torpid} use to show that tempering-based methods mix poorly.
Specifically, let $p_+=\N(b\bone_d, \sigma_+^2 I_d ) $, $p_-=\N(-b\bone_d, \sigma_-^2 I_d )$, and $w_+=w_-=\frac12$, with $b>0$ and $\sigma_+^2 > \sigma_-^2$.

\begin{cor}[Dimension-free ideal gap for $\alpha\sim d^{-1}$]
\label{cor:gaussian_mixture_dimension_free_gap}
Fix constants $b,\sigma_+,\sigma_-$ and set $\alpha=\kappa/d$ for a fixed $\kappa>0$. Then
\[
    \liminf_{d\to\infty}\, \mathrm{Gap}(P)
    \ge e^{-\kappa b^2}.
\]
\end{cor}

See the proof in Appendix~\ref{prf:gaussian_mixture_dimension_free_gap}.
Since $\alpha=e^{-T}$ along the OU process, taking $T=\log(d/\kappa)$ keeps the gap of $P$ bounded away from zero as $d\to\infty$.

By contrast, for this same unequal-variance mixture, \citet{woodard2009torpid} show that simulated and parallel tempering mix torpidly for any number and choice of temperatures, with a spectral gap that decays exponentially in $d$.
The cause is that tempered intermediate distributions $p^{\beta}$ ($\beta<1$) distort the relative mode weights, a phenomenon known as ``unfaithful''~\cite{neal1996sampling}, or ``low persistence''~\cite{woodard2009torpid}, or ``teleportation of mass''~\cite{mate2023learning}.
Concretely, at any inverse temperature $\beta<1$, the narrower mode $\N(-b\bone_d, \sigma_-^2 I_d)$ carries weight of order $(\sigma_-/\sigma_+)^{d(1-\beta)}$, exponentially small in $d$, so the probability of crossing from the wider mode to the narrower one is exponentially small at every intermediate level.

This contrast highlights the advantage of diffusion-path
interpolation, which preserves the mixture weights at every
intermediate time. For the Gaussian mixture above, the
asymptotic lower bound with $\alpha=\kappa/d$ depends on $b$ but not on the fixed component scales
$\sigma_+$ and $\sigma_-$.
Section~\ref{sec:experiment-mog} provides further numerical demonstration of this example.

\section{Metropolis adjustment for approximate diffusion paths}
\label{sec:algorithm}

In practice, the exact forward--backward diffusion processes described in the previous section are often intractable. 
In particular, the score function $\nabla\log p_t$ at intermediate time $t$ is typically unknown and needs to be approximated by learning-based methods; see Section~\ref{sec:practical}.
Moreover, the SDEs need to be solved numerically, which introduces discretization error. 
To correct for these errors, we propose a Metropolis--Hastings algorithm that uses the forward--backward simulation as a proposal and ensures the Markov chain has the correct invariant distribution.

We consider a general forward diffusion process $(X_t)_{t\in[0,T]}$ defined by the SDE
\begin{align}\label{equ:forward_sde}
    \rd X_t = f_t(X_t) \rd t + g_t \rd B_t,\quad X_0\sim p, \;\; t\in[0,T],
\end{align}
where $f:[0,T]\times\mathbb{R}^d \to \mathbb{R}^d$ and $g:[0,T] \to \mathbb{R}$ are drift and diffusion coefficients, respectively, and $B_t$ is a standard $d$-dimensional Brownian motion. 
Denote the marginal distribution of $X_t$ by $p_t$.
The time-reversal of the forward process is given by 
\begin{align}\label{equ:reverse_sde}
    \rd Y_t = [f_t(Y_t) - g_t^2 \nabla \log p_{t}(Y_t)] \widebar{\rd t} + g_t \rd \widebar B_t,\quad Y_T\sim p_T, \;\; t\in[T,0],
\end{align}
where $\widebar B_t$ is a standard Brownian motion when time flows backwards and $\widebar{\rd t}$ denotes an infinitesimal negative timestep.

In practice, $\nabla\log p_t$ is replaced with a learned score function $\hat s_t$ and the SDEs are simulated numerically. We describe the construction using Euler--Maruyama discretization.
Denote the backward drift as $-b_t(\bfy) = f_t(\bfy) - g_t^2 \hat s_t(\bfy)$. Let $\bfx_0\in\R^d$ be the current state of the Markov chain. We first discretize the forward SDE via Euler--Maruyama with step size $h$:
\begin{align*}
    \bfx_{k+1} = \bfx_{k} + f_{t_{k}}(\bfx_{k}) h + g_{t_{k}} \sqrt{h} \ep_k,\quad k=0,\ldots,N-1,
\end{align*}
where $t_k = kh$ and $\ep_k\iid\N(0, I_d)$.
Then, starting from $\bfy_N=\bfx_N$, we discretize the backward SDE as
\begin{align*}
    \bfy_{k} = \bfy_{k+1} + b_{t_k}(\bfy_{k+1}) h + g_{t_{k}} \sqrt{h} \zeta_k,\quad k=N-1,\ldots,0,
\end{align*}
where $\zeta_k\iid\N(0, I_d)$. Throughout, let $\varphi(\bfx; \bmu, \Sigma)$ denote the density of a multivariate Gaussian distribution with mean $\bmu$ and covariance matrix $\Sigma$, evaluated at $\bfx$. We treat $\bfy_0$ as the proposal and accept it with the Metropolis--Hastings acceptance probability
\begin{align}\label{equ:accept:prob}
    \frac{p(\bfy_0)}{p(\bfx_0) }\cdot \frac {\prod_{k=0}^{N-1}\leftvec{p}_{k}(\bfx_{k}\mid \bfx_{k+1})\cdot \prod_{k=0}^{N-1} \vec{p}_{k}(\bfy_{k+1}\mid \bfy_{k}) } {\prod_{k=0}^{N-1}\vec{p}_{k}(\bfx_{k+1}\mid \bfx_{k}) \cdot \prod_{k=0}^{N-1} \leftvec{p}_{k}(\bfy_{k}\mid \bfy_{k+1} ) } \wedge 1,
\end{align}
where 
\begin{align*}
    \vec{p}_{k}(\bfx_{k+1}\mid\bfx_{k}) &= \varphi(\bfx_{k+1}; \bfx_{k} + f_{t_{k}}(\bfx_{k}) h,\, g_{t_{k}}^2 h I_d),\\
    \leftvec{p}_{k}(\bfy_{k}\mid \bfy_{k+1}) &= \varphi(\bfy_{k}; \bfy_{k+1} + b_{t_{k}}(\bfy_{k+1}) h,\, g_{t_{k}}^2 h I_d).
\end{align*}
The resulting algorithm, named Metropolis-adjusted diffusion path (MAD-Path), is summarized in Algorithm~\ref{alg:mad-path} and illustrated in Figure~\ref{fig:illustration}.

\begin{algorithm}[]
\caption{One Iteration of MAD-Path}\label{alg:mad-path}
\begin{algorithmic}[1]
\Require {Target density $p_0$, learned score $\hat{s}_t$, drift $f$, diffusion $g$, time horizon $T$, number of steps $N$, current state $\mathbf{x}_0$}
\Ensure Next state $\bfx_0'$

\State $h \gets T / N$ 
\Statex \textit{Forward pass: run Euler--Maruyama on forward SDE}
\For{$k = 0, \ldots, N-1$}
    \State Sample $\varepsilon_k \sim \mathcal{N}(0, I_d)$
    \State $\mathbf{x}_{k+1} \gets \mathbf{x}_k + f_{t_k}(\mathbf{x}_k)\, h + g_{t_k}\sqrt{h}\, \varepsilon_k$ \Comment{$t_k = kh$}
\EndFor

\Statex \textit{Backward pass: run Euler--Maruyama on backward SDE}
\State $\mathbf{y}_N \gets \mathbf{x}_N$
\For{$k = N-1, \ldots, 0$}
    \State Sample $\zeta_k \sim \mathcal{N}(0, I_d)$
    \State $b_{t_k}(\mathbf{y}_{k+1}) \gets -f_{t_k}(\mathbf{y}_{k+1}) + g_{t_k}^2\, \hat{s}_{t_k}(\mathbf{y}_{k+1})$
    \State $\mathbf{y}_k \gets \mathbf{y}_{k+1} + b_{t_k}(\mathbf{y}_{k+1})\, h + g_{t_k}\sqrt{h}\, \zeta_k$
\EndFor

\Statex\textit{Metropolis--Hastings adjustment}
\State Set $\bfx_0' = \bfy_0$ with probability \eqref{equ:accept:prob}
and $\bfx_0' = \bfx_0$ otherwise.
\end{algorithmic}
\end{algorithm}

\begin{figure}[]
    \centering
    \includegraphics[width=\linewidth]{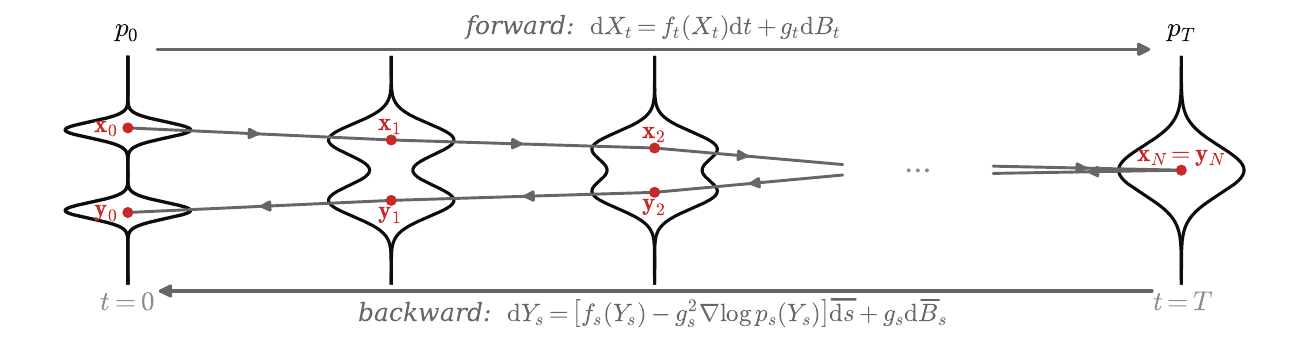}
    \caption{Illustration of MAD-Path (Algorithm~\ref{alg:mad-path}). From the current state $\bfx_0\sim p_0$, the forward process (top) generates $\bfx_1,\ldots,\bfx_N$ over $N$ discretization steps, transporting the multimodal target $p_0$ (left) toward the approximately Gaussian marginal $p_T$ (right). The backward process (bottom) starts from the shared endpoint $\bfy_N=\bfx_N$ and generates $\bfy_{N-1},\ldots,\bfy_0$, returning $\bfy_0$ as the proposed state. By smoothing the well-separated modes together near the noisy end ($t=T$), this forward--backward construction enables jumping between modes.}
    \label{fig:illustration}
\end{figure}

MAD-Path conceptually resembles Neal's tempered transition~\citep{neal1996sampling}, which likewise uses a sequence of intermediate distributions to interpolate between the target and an easier-to-sample distribution. Tempered transition first moves the candidate state forward to a ``higher temperature'', where the state space is easier to explore, and then moves it backward to the original ``lower temperature'', at which the target distribution is recovered.
The candidate is accepted with a Metropolis--Hastings acceptance probability that involves all intermediate states.
\citet{neal1996sampling} presents tempering as the canonical example of such interpolation, but notes that ``it may be appropriate to define a series of interpolating distributions in some other fashion, guided by whatever knowledge we may have of the properties of the distribution of interest.''
In this sense, our algorithm can be understood as a tempered transition guided by the forward--backward diffusion processes, which provides a more favorable interpolation than tempering, as discussed in Section~\ref{sec:ideal:transition}.

MAD-Path can also be interpreted as a transition on the augmented space of \sloppy{$\bfz:=(\bfx_0,\ldots,\bfx_N, \bfy_{N-1},\ldots,\bfy_0) $}, equipped with the joint distribution
\begin{align*}
    \Pi(\bfz) =p(\bfx_0)\cdot \prod_{k=0}^{N-1}\vec{p}_{k}(\bfx_{k+1}\mid \bfx_{k}) \cdot \prod_{k=0}^{N-1} \leftvec{p}_{k}(\bfy_{k}\mid \bfy_{k+1} ).
\end{align*}
From the current state $\bfx_0$, the forward and backward passes draw the remaining components $(\bfx_1,\ldots,\bfx_N,\bfy_{N-1},\ldots,\bfy_0)$ exactly from the conditional $\Pi(\cdot\mid \bfx_0)$. The sampler then proposes a move by applying the deterministic path-reversal map $\mathsf{R}: \R^{d(2N+1)} \to \R^{d(2N+1)}$,
\[
    \mathsf R(\bfx_0,\ldots,\bfx_N, \bfy_{N-1},\ldots,\bfy_0)
    :=(\bfy_0,\ldots,\bfy_{N-1},\bfx_N,
       \bfx_{N-1},\ldots,\bfx_0),
\]
which swaps the roles of the forward and backward paths. Because $\mathsf{R}$ is a volume-preserving involution, i.e. $\mathsf{R}\circ\mathsf{R}=\mathrm{id}$ and $\sfR$ preserves Lebesgue measure, the Metropolis acceptance probability is simply $\frac{\Pi(\mathsf{R}\bfz)}{\Pi(\bfz)} \wedge 1$, recovering Equation~\eqref{equ:accept:prob}.

The following theorem establishes that MAD-Path (Algorithm~\ref{alg:mad-path}) is reversible with respect to $p$, and hence admits $p$ as its invariant distribution. Notably, this holds \emph{regardless of the score and discretization errors}, which affect only the efficiency of the sampler, not its correctness.

\begin{thm}[Reversibility of MAD-Path transition kernel]
\label{thm:reversibility}
The MAD-Path transition kernel is reversible with respect to $p$.
\end{thm}
\begin{proof}[Proof of Theorem~\ref{thm:reversibility}]
For two measurable sets $A,B\subset\R^d$, the transition probability from $A$ to $B$ can be written as
\begin{align}\label{equ:reversible:prf}
    \int \Indc{\bfx_0\in A,\bfy_0\in B} \Pi(\bfz) \cdot \Big(\frac{\Pi(\mathsf{R}\bfz)}{\Pi(\bfz)} \wedge 1 \Big) \rd \bfz 
    + \int \Indc{\bfx_0\in A\cap B } \Pi(\bfz) \cdot \Big(1 - \frac{\Pi(\mathsf{R}\bfz)}{\Pi(\bfz)} \wedge 1 \Big) \rd \bfz.
\end{align}
Writing $\bfz_0=\bfx_0$ and $\bfz_{-1}=\bfy_0$ for the first and last components of $\bfz$, the first term equals
\begin{align*}
    \int \Indc{\bfx_0\in A,\bfy_0\in B} \Pi(\bfz) \cdot \Big(\frac{\Pi(\mathsf{R}\bfz)}{\Pi(\bfz)} \wedge 1 \Big) \rd \bfz &= 
    \int \Indc{\bfz_0\in A,\bfz_{-1}\in B}\cdot \big(\Pi(\bfz) \wedge \Pi(\mathsf{R}\bfz) \big) \rd \bfz \\
    &=\int \Indc{\bfw_{-1}\in A,\bfw_{0}\in B}\cdot \big(\Pi(\mathsf{R}\bfw) \wedge \Pi(\bfw) \big) \rd \bfw \\
    &=\int \Indc{\bfw_{0}\in B,\bfw_{-1}\in A}\cdot \big(\Pi(\bfw) \wedge \Pi(\mathsf{R}\bfw) \big) \rd \bfw ,
\end{align*}
where the second equality applies the change of variable $\bfw=\mathsf{R}\bfz$ and uses that $\mathsf{R}$ is a volume-preserving involution. Thus, the last display is symmetric in $A$ and $B$. The second term in~\eqref{equ:reversible:prf} is also symmetric in $A$ and $B$ by definition. Hence, the transition kernel is reversible.

\end{proof}

This result is not limited to Euler--Maruyama discretization. For other numerical integrators, provided that $\vec p_k$ and $\leftvec p_k$ in~\eqref{equ:accept:prob} are the densities of the transitions actually used to generate the paths, the transition kernel remains reversible.

\section{Analysis of acceptance probability and tuning guidelines}
\label{sec:analysis}

Although MAD-Path (Algorithm~\ref{alg:mad-path}) is guaranteed to have the correct invariant distribution, the acceptance probability is affected by discretization error and score-estimation error. We now analyze how these errors affect the acceptance probability~\eqref{equ:accept:prob}. 

\subsection{Effects of discretization error}

We first isolate the effect of time discretization by assuming that the exact
score
\[
    s_t(\bfx):=\nabla\log p_t(\bfx), \qquad t\in[0,T]
\]
is available. 
Denote the log Metropolis--Hastings ratio with step size $h$ by
\begin{align}\label{equ:def:Delta}
    \Delta_{h}
    :=\log \left[\frac{p(\bfy_0)}{p(\bfx_0)}\cdot \frac {\prod_{k=0}^{N-1}\leftvec{p}_{k}(\bfx_{k}\mid \bfx_{k+1})\cdot \prod_{k=0}^{N-1} \vec{p}_{k}(\bfy_{k+1}\mid \bfy_{k}) } {\prod_{k=0}^{N-1}\vec{p}_{k}(\bfx_{k+1}\mid \bfx_{k}) \cdot \prod_{k=0}^{N-1} \leftvec{p}_{k}(\bfy_{k}\mid \bfy_{k+1} ) } \right].
\end{align}
Thus, the acceptance probability is $1\wedge e^{\Delta_h}$.
The analysis below concerns the Euler--Maruyama scheme. Remark~\ref{rmk:exponential_integrator} gives the corresponding
result for an exponential integrator.

Throughout this section, the forward marginals $p_t$ are strictly positive and satisfy the Fokker--Planck equation $\partial_t p_t=-\nabla\cdot(f_tp_t)+\tfrac12\Delta(g_t^2p_t)$; here $\partial_t$ denotes the time derivative, while $\nabla$, $\nabla\cdot$, $\Delta$, and all higher-order operators act on the spatial variable $\bfx$. 
We write $L_t(\bfx):=\log p_t(\bfx)$, so that $s_t=\nabla L_t$, and let $Y$ denote the exact reverse diffusion initialized at $Y_T=X_T$ and driven, conditionally on $X_T$, by reverse-time Brownian motion independent of the forward path. 
We assume that the Jacobian $\nabla f_t$ is symmetric for every $t$, as holds whenever $f_t$ is a gradient field.
For integers $r,s\ge0$, let $C^{r,s}$ denote the space of (scalar- or $\R^d$-valued) functions on $[0,T]\times\R^d$ whose mixed derivatives up to order $r$ in time and order $s$ in space are continuous and grow at most polynomially in $\bfx$, uniformly in $t$.
We impose the following regularity conditions on the coefficients and the Euler--Maruyama scheme.

\begin{assumption}[Regularity of the exact forward--backward diffusion]\label{assump:regularity}
There exists $h_0>0$ such that the following holds.
\begin{enumerate}[(i)]
    \item \emph{Coefficients.} $g\in C^1([0,T])$ is bounded away from zero and infinity.
    \item \emph{Smoothness and growth.} $f\in C^{1,3}$ and $L\in C^{2,4}$.
    \item \emph{Well-posedness and moments.} The forward and reverse SDEs have unique solutions, and for every $q\ge2$ the solutions $X_t,Y_t$ and the Euler iterates $\bfx_k,\bfy_k$ have $L^q$ norms bounded uniformly in $t$, $k$, and $0<h\le h_0$.
    \item \emph{Strong error.} For every $q\ge2$, the Euler scheme can be coupled to the exact solution through shared Brownian increments so that
    \begin{equation}
    \label{eq:exact-score-strong-error}
        \max_{0\leq k\leq N}
        \left(
            \|\bfx_k-X_{t_k}\|_{L^q}
            +
            \|\bfy_k-Y_{t_k}\|_{L^q}
        \right)
        \leq C_q\,h^{1/2},
        \qquad 0<h\leq h_0.
    \end{equation}
\end{enumerate}
Conditions (iii)--(iv) hold under global Lipschitz and linear-growth assumptions on the forward and reverse drifts together with (i)--(ii); see, for example, \cite[Chapter 10]{kloeden1992numerical}. 
Some of these conditions may be stronger than necessary; we adopt them to keep the analysis self-contained and to avoid technical digressions. 
\end{assumption}

Define 
\begin{equation*}
    \mathcal D(T)
    :=
    4\int_0^T
    \mathbb E\!\left[
        \left\|\nabla f_t(X_t)
        -\frac{g_t^2}{2}\nabla s_t(X_t)\right\|_{\mathrm F}^2
    \right]\,\rd t,
\end{equation*}
where the expectation is taken over $(X_t)_{t\in[0,T]}$ satisfying the SDE~\eqref{equ:forward_sde}.
In the special case of the OU process~\eqref{equ:ou}, where $f_t(\bfx)=-\tfrac12\bfx$ and $g_t=1$, this reduces to $\calD(T)=\int_0^T\EE{\|\nabla s_t(X_t)+I_d\|_{\mathrm F}^2}\,\rd t$, the accumulated squared deviation of the log-density Hessian $\nabla s_t=\nabla^2\log p_t$ from $-I_d$, its value at the Gaussian equilibrium $\N(0,I_d)$.

\begin{thm}[Effects of discretization error on acceptance probability]
\label{thm:exact-score-discretization}
Suppose the exact score $s_t$ is available for $t\in[0, T]$.
Under Assumption~\ref{assump:regularity}, and under stationarity $\bfx_0\sim p_0$, as $h\to0$,
\begin{align*}
    \mathbb E[\Delta_h] &= -\frac{h}{2}\mathcal D(T)+O(h^{3/2}),\\
    \operatorname{Var}(\Delta_h) &= h\mathcal D(T)+O(h^{3/2}).
\end{align*}
\end{thm}
The proof is given in Appendix~\ref{prf:exact-score-discretization}.
In particular, the variance of the log Metropolis--Hastings ratio is linear in $h$ when $\calD(T)>0$, while its mean is asymptotically minus one-half of that variance.
Consequently, the rejection probability satisfies
\[
1-\EE{e^{\Delta_h}\wedge 1}\leq \EE{|\Delta_h|} = O(h^{1/2}).
\]
In the continuous-time limit $h\downarrow 0$, the discretized proposal converges to the ideal diffusion-path kernel defined in Section~\ref{sec:ideal:transition}, which has acceptance probability $1$.
The discretization error therefore induces a rejection probability that is at most $O(h^{1/2})$ and a trade-off between acceptance and computational cost: smaller $h$ raises the acceptance probability but requires more discretization steps, and hence higher cost per proposal.

\begin{rmk}[High-dimensional scaling for product i.i.d.~targets]
\label{rem:exact-score-product-scaling}
To understand the trade-off between acceptance probability and computational cost, we adopt the high-dimensional optimal-scaling framework for product-form targets~\citep{roberts1997weak,roberts2001optimal}. We restrict attention to product-form targets for analytical tractability, although analogous arguments can be extended to more general target distributions \citep{Yang2020}. Specifically, following the same setting as in \cite{beskos2013optimal}, we assume that the target consists of $n$ independent copies of the $d$-dimensional component analyzed above, so the total dimension is $nd$.
For each $n$, the log Metropolis--Hastings ratio decomposes as
\[
    \Delta_n=\sum_{j=1}^n\Delta_{n,j},
\]
where each $\Delta_{n,j}$ is as defined in Equation~\eqref{equ:def:Delta}, and under stationarity, are i.i.d. for $1\leq j\leq n$. 
By Theorem~\ref{thm:exact-score-discretization}, each component has mean and variance of order $h$. Thus, a non-degenerate limit for $\Delta_n$ requires $h\asymp n^{-1}$. Fix $T$, and let $h_n=\ell/n$, $N_n=\lceil Tn/\ell\rceil$. 
By Theorem~\ref{thm:exact-score-discretization} and the Lyapunov central limit theorem, as $n\to\infty$ we have
\[
    \Delta_n
    \ \Rightarrow\
    \N\Big(-\frac{\ell\mathcal D(T)}{2}, \ell\mathcal D(T)\Big),
\]
and hence the acceptance probability has the limit
\[
    \mathbb E[1\wedge e^{\Delta_n}]
    \to
    a(\ell)
    :=
    2\Phi\Big(
        -\frac{\sqrt{\ell\mathcal D(T)}}{2}
    \Big),
\]
where $\Phi$ is the standard normal cumulative distribution function. Since the cost of a proposal is proportional to $N_n\asymp n/\ell$, the limiting cost-normalized acceptance probability is proportional to $\ell\,\Phi(-\sqrt{\ell\mathcal D(T)}/2)$. 
At the maximizer $\ell^*=\arg\max_{\ell}\ell\,\Phi(-\sqrt{\ell\mathcal D(T)}/2)$, the acceptance probability is approximately $0.234$. 
Figure~\ref{fig:optimal-scaling} illustrates this scaling numerically for a $500$-dimensional i.i.d.\ Gaussian target with exact score, where the empirically optimal acceptance rate is close to $0.234$.

The value $0.234$ coincides numerically with the classical optimum for random-walk Metropolis~\citep{roberts1997weak}. This coincidence does not, however, imply that the proposed algorithm exhibits random-walk behavior; the fixed-$T$ forward--backward construction remains a global proposal. 
The scaling $h_n\asymp n^{-1}$ is needed solely to control the accumulation of discretization error across the $n$ independent components, each contributing a perturbation of variance $O(h)$ to the log-ratio.
For comparison, in HMC each component contributes a perturbation of variance $O(h^4)$, giving the $h_n\asymp n^{-1/4}$ scaling~\citep{beskos2013optimal}.
The more favorable step-size scaling of HMC arises because it discretizes deterministic Hamiltonian dynamics with the second-order leapfrog integrator, which is more accurate than the Euler--Maruyama scheme used here. 
\end{rmk}

\begin{figure}[htbp]
    \centering
    \includegraphics[width=0.5\linewidth]{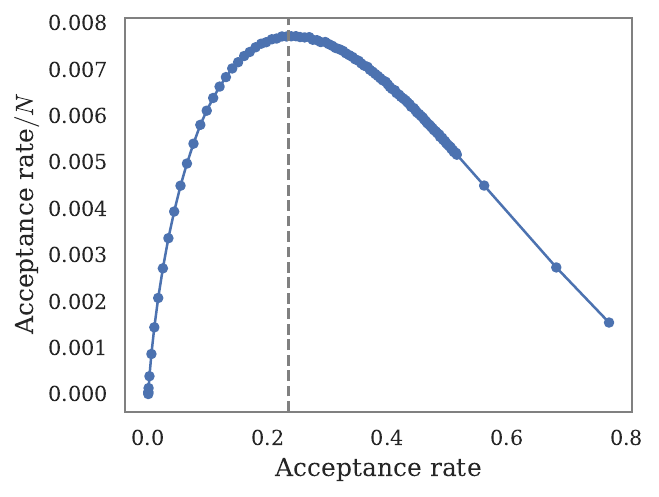}
    \caption{Cost-normalized efficiency versus acceptance rate for the target $\N(10\cdot\boldsymbol{1}_d,\,3^2 I_d)$ in dimension $d=500$, with exact score and horizon $T=2$, as the step size $h$ is varied. The $x$-axis is the empirical acceptance rate; the $y$-axis is the acceptance rate divided by the number of discretization steps $N$ per proposal. The vertical line marks $0.234$, the acceptance rate that maximizes efficiency in the high-dimensional limit (Remark~\ref{rem:exact-score-product-scaling}).}
    \label{fig:optimal-scaling}
\end{figure}

\subsection{Effects of score-estimation error}
\label{subsec:score-error}

We now consider the case where the score function $s_t=\nabla\log p_t$ is not known exactly, but is approximated by $\widehat s_t$. 
Let $\widehat Y$ denote the continuous-time reverse diffusion driven by the
approximate score,
\begin{align}\label{equ:reverse_sde_approx}
    \rd \widehat Y_t = \big[f_t(\widehat Y_t) - g_t^2 \widehat s_{t}(\widehat Y_t)\big]\, \widebar{\rd t} + g_t \rd \widebar B_t,\quad \widehat Y_T\sim p_T, \;\; t\in[T,0].
\end{align}
Thus, the backward Euler scheme $\{\bfy_k\}_{k=0}^N$ is the time discretization of $\widehat Y$. Define the score error $\delta_t:=\widehat s_t-s_t$.

\begin{assumption}[Regularity of the score error and estimated reverse path]
\label{assump:regularity:score}
Assume that $\delta\in C^{1,3}$ and that the approximate reverse diffusion~\eqref{equ:reverse_sde_approx} is well-posed. Its Euler scheme $\{\bfy_k\}_{k=0}^N$ satisfies the uniform $L^q$ moment and strong-error conditions in~\eqref{eq:exact-score-strong-error}, with $Y_{t_k}$ replaced by $\widehat Y_{t_k}$. Moreover, the forward and approximate-reverse Euler schemes have weak order one, uniformly over the time grid.
\end{assumption}

Define the integrated score error
\begin{equation}
\label{eq:integrated-score-error}
    \mathcal E_{\mathrm{score}}(\delta,T)
    :=
    \frac12\int_0^T g_t^2
    \left\{
        \mathbb E\!\left[\|\delta_t(X_t)\|^2\right]
        +
        \mathbb E\!\left[\|\delta_t(\widehat Y_t)\|^2\right]
    \right\}\,\rd t.
\end{equation}
Let $\bbP_Y$ denote the path law of the exact reverse diffusion~\eqref{equ:reverse_sde} and $\bbP_{\widehat Y}$ that of the approximate reverse diffusion~\eqref{equ:reverse_sde_approx}, both initialized from $p_T$. 
Under regularity conditions, the two directed Kullback--Leibler divergences between these path laws are
\begin{align*}
    \mathrm{KL}(\bbP_{Y} \,\|\, \bbP_{\widehat Y})
    &=\frac12\int_0^T g_t^2
      \mathbb E \left[\|\delta_t(Y_t)\|^2\right]\,\rd t,\\
    \mathrm{KL}(\bbP_{\widehat Y} \,\|\, \bbP_{Y})
    &=\frac12\int_0^T g_t^2
      \mathbb E\left[\|\delta_t(\widehat Y_t)\|^2\right]\,\rd t.
\end{align*}
Therefore, we can write 
\begin{align*}
\mathcal E_{\mathrm{score}}(\delta,T) = \mathrm{KL}(\bbP_{Y} \,\|\, \bbP_{\widehat Y}) + \mathrm{KL}(\bbP_{\widehat Y} \,\|\, \bbP_{Y}),
\end{align*}
the symmetrized Kullback--Leibler divergence between the exact and approximate path laws.

\begin{thm}[Effect of score-estimation error]
\label{thm:score-error-mean}
Under Assumptions~\ref{assump:regularity}--\ref{assump:regularity:score}, and under stationarity $\bfx_0\sim p_0$,
\begin{equation*}
    \mathbb E[\Delta_h]
    =-\mathcal E_{\mathrm{score}}(\delta,T)+O(h),
    \qquad h\to0.
\end{equation*}
Consequently, the acceptance probability satisfies
\begin{align*}
    \EE{1\wedge e^{\Delta_h}} \geq \frac12 \exp\big(-\mathcal E_{\mathrm{score}}(\delta,T)+O(h)\big).
\end{align*}
\end{thm}
The proof is given in Appendix~\ref{prf:thm:score-error-mean}.
As $h\to0$, the mean log-ratio converges to $-\mathcal E_{\mathrm{score}}(\delta,T)$, a limit set by the score error. Consequently, once the score error dominates the discretization error, shrinking $h$ barely affects the acceptance probability. Improving acceptance then requires reducing $\mathcal E_{\mathrm{score}}(\delta,T)$ itself, either by shortening the diffusion horizon $T$ or by learning a more accurate score $\widehat s_t$. We turn to the latter next.

\section{Score estimation in practice}
\label{sec:practical}

MAD-Path (Algorithm~\ref{alg:mad-path}) requires an approximate score
$\widehat s_t\approx\nabla\log p_t$, whose accuracy governs the sampler's
efficiency (Theorem~\ref{thm:score-error-mean}). We first describe a
path-space variational formulation that learns these scores from an
unnormalized target density without requiring samples from $p$, and then
discuss its implementation for multimodal targets and connections to
related work.

\subsection{Path-space variational score estimation}

Let $p^{\mathrm{ref}}$ be a tractable reference distribution, and let
$p_t^{\mathrm{ref}}$ denote the law of $X_t$ under the noising diffusion
\begin{align}\label{equ:reference:process}
    \rd X_t
    =
    f_t(X_t)\,\rd t + g_t\,\rd B_t,
    \qquad
    X_0\sim p^{\mathrm{ref}},
    \qquad t\in[0,T].
\end{align}
Suppose that the reference score
$s_t^{\mathrm{ref}}:=\nabla\log p_t^{\mathrm{ref}}$
is available in closed form, as is the case when $f_t$ is linear and
$p^{\mathrm{ref}}$ is Gaussian or a Gaussian mixture. The time reversal
of~\eqref{equ:reference:process} is
\begin{align*}
    \rd Y_t
    =
    \bigl[
        f_t(Y_t)-g_t^2s_t^{\mathrm{ref}}(Y_t)
    \bigr]\,\widebar{\rd t}
    +
    g_t\,\rd\widebar B_t,
    \qquad
    Y_T\sim p_T^{\mathrm{ref}},
    \qquad t\in[T,0],
\end{align*}
with the same convention as in~\eqref{equ:reverse_sde}.

Let $\bbP^{\mathrm{ref}}$ denote the path measure of this reverse
process. We define the target path measure $\bbP^*$ by
\begin{align*}
    \frac{\rd\bbP^*}{\rd\bbP^{\mathrm{ref}}}(Y)
    =
    \frac{p(Y_0)}{p^{\mathrm{ref}}(Y_0)}.
\end{align*}
Thus, $\bbP^*$ is the path law of the same forward diffusion initialized
from $p$, expressed in reverse-time coordinates; in particular,
$\bbP_t^*=p_t$ for every $t\in[0,T]$.

We approximate $\bbP^*$ using the controlled SDE
\begin{align*}
    \rd Y_t^u
    =
    \Bigl[
        f_t(Y_t^u)
        -
        g_t^2
        \bigl(
            s_t^{\mathrm{ref}}(Y_t^u)+u_t(Y_t^u)
        \bigr)
    \Bigr]\,\widebar{\rd t}
    +
    g_t\,\rd\widebar B_t,
    \qquad
    Y_T^u\sim p_T^{\mathrm{ref}},
    \qquad t\in[T,0],
\end{align*}
where $u$ is a control and
$\widehat s_t:=s_t^{\mathrm{ref}}+u_t$ is the approximate score used by
MAD-Path. Let $\bbQ^u$ denote the resulting path measure. The control is
trained by minimizing $\mathrm{KL}(\bbQ^u\,\|\,\bbP^*)$.

By Girsanov's theorem,
\begin{equation}\label{equ:reverse:path:kl}
\begin{aligned}
    \mathrm{KL}(\bbQ^u\,\|\,\bbP^*)
    &=
    \mathrm{KL}(\bbQ^u\,\|\,\bbP^{\mathrm{ref}})
    +
    \bbE_{\bbQ^u}
    \left[
        \log
        \frac{p^{\mathrm{ref}}(Y_0^u)}{p(Y_0^u)}
    \right]
    \\
    &=
    \bbE_{\bbQ^u}
    \left[
        \int_0^T
        \frac{g_t^2}{2}
        \|u_t(Y_t^u)\|^2
        \,\rd t
        +
        \log
        \frac{p^{\mathrm{ref}}(Y_0^u)}{p(Y_0^u)}
    \right].
\end{aligned}
\end{equation}
This is a stochastic optimal control problem with running cost
$\tfrac{g_t^2}{2}\|u_t\|^2$ and terminal cost
$\log\frac{p^{\mathrm{ref}}(Y_0^u)}{p(Y_0^u)}$. 

Writing $s_t:=\nabla\log p_t$, the chain rule for KL divergence and
another application of Girsanov's theorem give
\begin{align*}
    \mathrm{KL}(\bbQ^u\,\|\,\bbP^*)
    &=
    \mathrm{KL}
    \bigl(
        p_T^{\mathrm{ref}}
        \,\|\,
        p_T
    \bigr)+
    \bbE_{\bbQ^u}
    \left[
        \int_0^T
        \frac{g_t^2}{2}
        \bigl\|
            s_t^{\mathrm{ref}}(Y_t^u)
            +
            u_t(Y_t^u)
            -
            s_t(Y_t^u)
        \bigr\|^2
        \,\rd t
    \right].
\end{align*}
Therefore, under standard regularity conditions, if the control class
contains
\begin{align*}
    u_t^*
    =
    \nabla\log p_t-\nabla\log p_t^{\mathrm{ref}},
    \qquad t\in[0,T],
\end{align*}
then the optimal control recovers the exact score, and the minimum
objective value is
$\kl{p_T^{\mathrm{ref}}}{p_T}$. Importantly, this conclusion holds for
every finite $T$: terminal matching changes the minimum value of the
objective but not its minimizer.

Standard diffusion samplers commonly initialize the reverse process from $p_T^{\mathrm{ref}}$. Even with the exact score and continuous-time simulation, the resulting path-space KL divergence contains the irreducible terminal-mismatch term $\kl{p_T^{\mathrm{ref}}}{p_T}$. Eliminating this term requires the memoryless condition $p_T^{\mathrm{ref}}=p_T$~\citep{zhang2022path,vargasdenoising}, which, for a stationary OU reference, is generally attained only as $T\to\infty$. MAD-Path instead initializes the reverse pass at the endpoint of its forward pass, while the Metropolis--Hastings correction preserves the target for every finite $T$. The horizon can therefore be chosen based on mode overlap (Section~\ref{sec:ideal:transition}), acceptance probability (Section~\ref{sec:analysis}), and computational cost.

Although the population objective~\eqref{equ:reverse:path:kl} recovers the exact score when the control class is sufficiently rich, its parametrized optimization can be difficult for multimodal targets. In particular, the reverse path-space KL is mode-seeking: it penalizes placing mass in low-density regions of the target more severely than omitting well-separated modes, and can therefore lead to mode collapse~\citep{vargas2024transport,he2025no}. A common strategy in multimodal MCMC is to supplement the sampler with a mode-discovery procedure~\citep{pompe2020framework,tawn2020annealed,tan2025accelerate,latuszynski2025mcmc}. We adopt the same principle for score estimation. We first locate the dominant modes and then construct $p^{\mathrm{ref}}$ as a Gaussian mixture whose components encode them. This provides the controlled process with coverage of the relevant modes while retaining a closed-form reference score under linear noising dynamics. Details are given in Appendix~\ref{appendix:score}; a related strategy is used in~\cite{noble2025learned}.

\subsection{Related work}

Different choices of the reference process~\eqref{equ:reference:process}
recover or closely relate to several existing diffusion-based samplers.
The path integral sampler~\citep[PIS;][]{zhang2022path} uses Brownian
motion ($f_t=0$, $g_t=1$) pinned at the origin, while the denoising
diffusion sampler~\citep[DDS;][]{vargasdenoising} uses an equilibrium OU
process with $p^{\mathrm{ref}}=\N(0,I_d)$ and $T\to\infty$. Non-equilibrium annealed
references are studied in~\cite{albergo2025nets,choi2026non}. Other
approaches replace the path-space objective with regression targets
derived from the target score identity, enabling simulation-free or
adjoint-based training
\citep{de2024target,havens2025adjoint,blessing2026bridge}, or estimate
the intermediate scores through Tweedie's formula using importance
sampling or MCMC
\citep{huang2024diffusion,grenioux2024stochastic}.

A separate line of work learns the intermediate log-densities along the
diffusion path
\citep{mate2023learning,sun2024dynamical,chemseddine2025neural,phillips2024particle},
which enables corrections based on annealed importance sampling or
sequential Monte Carlo, as in~\cite{phillips2024particle}. However,
\citet{grenioux2026diffusion} show that such estimates can be
mode-blind: even when the corresponding score is accurate, the
estimated density may misrepresent the relative weights of separated
modes. MAD-Path requires only the learned score, not intermediate
log-density estimates, and its Metropolis--Hastings correction preserves
the target's mode weights. Section~\ref{sec:experiments} provides a
numerical illustration.

MAD-Path is therefore agnostic to how the score and reference are
constructed: any of the approaches above may be used within the
algorithm. The choice affects acceptance and efficiency, but not the
invariant distribution. Nevertheless, an informative reference and an
accurate score remain crucial for efficient sampling, particularly for
complex multimodal targets.

\section{Numerical experiments} \label{sec:experiments}

In this section, we present numerical experiments for our method (Algorithm~\ref{alg:mad-path}).
We evaluate its performance over six examples: a mixture of Gaussians with unequal variances; a mixture of well-separated skew-normal distributions; two Bayesian mixture models applied to real-world datasets: the Old Faithful geyser \cite{azzalini1990look} and fetal deaths in litters \cite{brooks1997finite}; a Bayesian model for sensor network localization \cite{ihler2004nonparametric}; and a seemingly unrelated regression (SUR) model \cite{zellner1962efficient} with multimodal profile likelihood.

On the first two synthetic targets, we compare MAD-Path against six baselines: adaptive parallel tempering \cite[APT;][]{miasojedow2013adaptive}, adaptive parallel tempering with deterministic even-odd swaps \cite[APT-DEO;][]{okabe2001replica}, tempered transition \cite[TT;][]{neal1996sampling}, annealed importance sampling \cite[AIS;][]{neal2001annealed}, sequential Monte Carlo \cite[SMC;][]{gordon1993novel,del2006sequential}, as well as the unadjusted diffusion (UD) sampler, which uses the same learned score $\widehat s_t$ as MAD-Path but omits the Metropolis--Hastings correction.
All the tempering-based methods use ten MALA steps for each within-temperature transition. 
We run APT and APT-DEO such that the swap rate, defined as the percentage of proposed state exchanges between parallel temperature chains that are successfully accepted, reaches around its theoretically optimal value of 0.234~\cite{atchade2011towards}.
For TT, we calibrate the geometric inverse-temperature ladder and the per-level MALA step size to achieve a moderate acceptance rate.
For AIS and SMC, we use the adaptive schedule of~\cite{zhou2016toward}. 
In Section~\ref{sec:fetal:deaths}, we examine the sensitivity of MAD-Path to the choice of horizon $T$ and number of steps $N$. 
In the remaining examples, we mainly compare MAD-Path against APT-DEO, which proved the strongest competitor among the baselines. MAD-Path uses the Euler--Maruyama integrator in all experiments.

We measure the computational cost of each method by the number of target score evaluations and ensure each method uses roughly the same number of evaluations. 
MAD-Path is an amortized sampler: the target score is evaluated only during training, and the trained score $\widehat s_t$, which does not involve the target score, is used during MCMC sampling. Therefore, the cost of gradient evaluation of MAD-Path only comes from the training phase, while the other MCMC methods (with MALA as the within-temperature transition) require target score evaluations at every step. 
The total wall clock time, including both training and sampling, is also comparable, as listed in Section~\ref{app:experiments:details}.

All experiments were conducted on an Apple M3 machine with 16 GB of unified memory. Further implementation details, including the choices of $p^{\mathrm{ref}}$, $T$, and $h$, as well as further results and details, are deferred to Section~\ref{appendix:experiments}. 
The code to reproduce the results is available at \href{https://github.com/chenh77/mad-path.git}{\tt github.com/chenh77/mad-path}.

\subsection{Mixture of Gaussians with unequal variances} \label{sec:experiment-mog}

We first evaluate our algorithm on the two-Gaussian mixture target with unequal variances,
\begin{equation*}
    p(\bfx) = \frac12 \, \varphi\big(\bfx;\, \boldsymbol{1}_d, \sigma_+^2 I_d\big) + \frac12 \, \varphi\big(\bfx;\, - \boldsymbol{1}_d, \sigma_-^2 I_d\big),
\end{equation*}
with $d=20$, $\sigma_+^2=0.5$, and $\sigma_-^2=0.2$. As discussed in Section~\ref{sec:ideal:transition}, tempering-based methods are prone to torpid mixing on such targets~\cite{woodard2009torpid}.

We compare MAD-Path against tempering-based samplers under a matched computational budget, measured by the number of target score evaluations (MAD-Path: $5.0\times 10^5$; APT and APT-DEO: $2.5\times 10^6$; TT: $3.5\times 10^6$; AIS and SMC: $7.0\times 10^5$). We run MAD-Path for 10,000 iterations, at an average acceptance rate of 0.372. 
We run APT, APT-DEO, and TT for 12,500 iterations each, discarding the first 2,500 as burn-in. 
The average swap acceptance rates of APT and APT-DEO are 0.387 and 0.207, respectively, and the average acceptance rate of TT is 0.356.
AIS and SMC each use 10,000 particles, for which we report the histograms and effective sample size (ESS) of the weighted samples. Finally, we include the unadjusted reverse-diffusion sampler as a baseline.

Figure~\ref{fig:mog-particle-trace} reports histograms of the projection $\bfv_1^\top \bfx$ (top row) and the corresponding trace plots (bottom row), where $\bfv_1 = \boldsymbol{1}_d/\sqrt{d}$ is the unit vector along which the two modes are most separated. MAD-Path is the only sampler to recover the target accurately (red dashed), and its trace exhibits the most frequent crossings between the two modes. APT-DEO also moves between the modes but overrepresents the narrow one over the finite run. Because the narrow mode has a much higher peak density, swaps that replace a cold narrow-mode state with a hotter wide-mode state are rarely accepted, resulting in long residence times in the narrow mode.
APT and TT each remain trapped in the wider mode, even under extensive tuning of the temperature schedule. AIS likewise yields biased mode-weight estimates and an ESS of only 17; the resampling step in SMC improves the ESS but still fails to recover the target accurately. Finally, UD is biased, underrepresenting the narrower mode, whereas MAD-Path's Metropolis--Hastings correction recovers the target's mode weights accurately.

\begin{figure}[htbp]
    \centering
    \includegraphics[width=\linewidth]{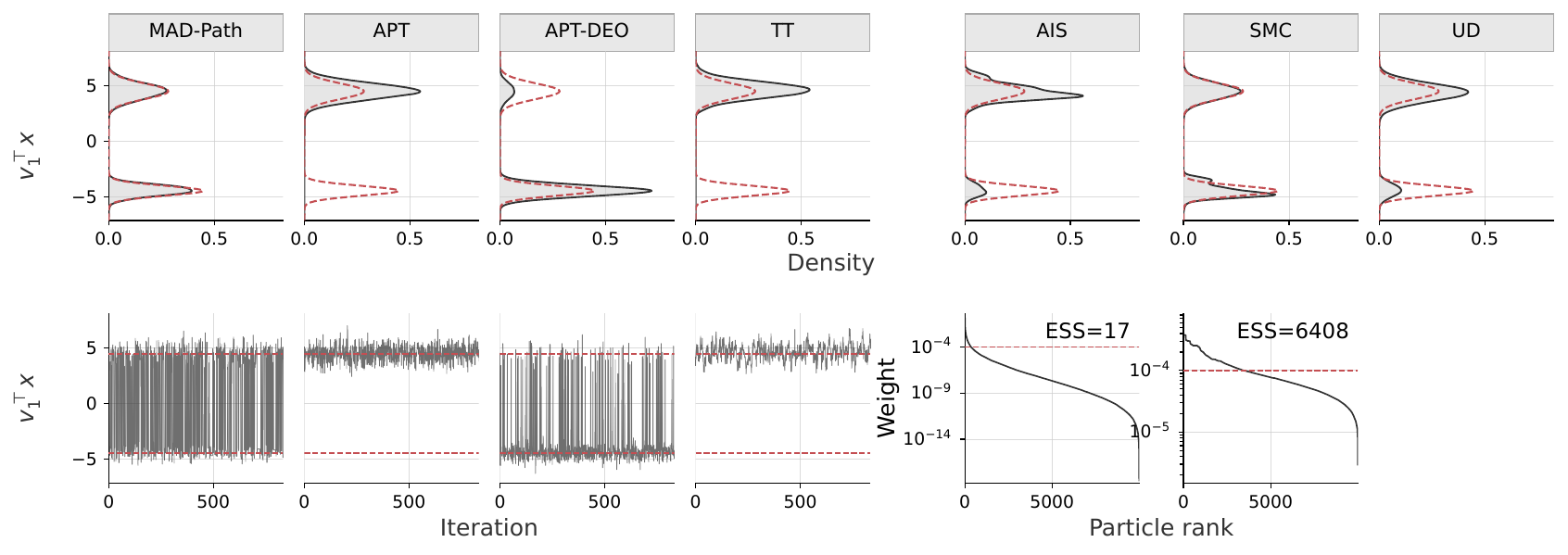}
    \caption{Results for the two-Gaussian mixture example. Top row: empirical distribution of the diagonal projection $\bfv_1^\top \bfx$ for all seven samplers, with red dashed lines indicating the true marginal. Bottom row: trace plots of $\bfv_1^\top \bfx$ for the four MCMC samplers and sorted-weight curves for AIS and SMC.}
    \label{fig:mog-particle-trace}
\end{figure}

\subsection{Mixture of skew normals} \label{sec:experiment-skew}

We next consider a mixture of thirteen well-separated, 20-dimensional skew-normal distributions, following \cite{tawn2020annealed,tan2025accelerate}, who extend a two-dimensional example of \cite{tjelmeland2001mode}:
\begin{align*}
    p(\bfx) \propto \sum_{k=1}^{13} w_k \prod_{j=1}^{20} \varphi\!\left(\frac{x_j-(\boldsymbol{\mu}_k)_j}{\sigma_k}\right)\Phi\!\left(\alpha \, \frac{x_j-(\boldsymbol{\mu}_k)_j}{\sigma_k}\right),
\end{align*}
with equal weights $w_k=1/13$, skewness $\alpha = 5$, and scale $\sigma_k=1$ for all $k$. The location vectors $\boldsymbol{\mu}_k$ fix the first two coordinates as in \cite{tjelmeland2001mode} and repeat this pattern across the remaining eighteen; Figure~\ref{fig:skew-particle-trace} shows the resulting mode locations in the first two coordinates.

We compare the same seven samplers as in Section~\ref{sec:experiment-mog}, each using a comparable number of gradient evaluations ($1$--$2\times 10^6$).
All MCMC samplers are run for 10,000 iterations, and AIS and SMC use 10,000 particles. MAD-Path attains an average acceptance rate of 0.240 and TT of 0.362. The average swap acceptance rates of APT and APT-DEO are 0.561 and 0.192, respectively. 
Figure~\ref{fig:skew-particle-trace} shows the histogram of the first coordinate together with scatter and trace plots of the first two coordinates. Table~\ref{tab:skew} reports the accuracy of the estimated mode weights: the total variation (TV) distance to the true weights (all equal) and the ratio of the largest to smallest estimated weight (Max/Min); the ratio is marked NA when a method fails to recover all modes. Results are averaged over 100 independent runs, with standard deviations in parentheses. In each run, MAD-Path is initialized from a draw from the learned reverse diffusion, while all other methods are initialized from the Gaussian distribution.

A few observations are in order.
\begin{itemize}
    \item Table~\ref{tab:skew} shows that MAD-Path, APT-DEO and UD all recover the 13 modes consistently across the 100 replicates. Among these, MAD-Path estimates the mode weights substantially more accurately than the other two, by both the TV distance to the true (uniform) weights and the Max/Min imbalance ratio.

    \item The histogram in Figure~\ref{fig:skew-particle-trace} shows that MAD-Path closely matches the true marginal of the first coordinate. Notably, the unadjusted reverse diffusion (last column) covers all local modes, yet its samples remain approximately Gaussian within each mode, as seen in both the histogram and the scatter plot: the estimated score is too inaccurate to capture the local skewness. MAD-Path corrects for this through the Metropolis--Hastings step, producing samples that reflect the true local shape.

    \item The two-dimensional trace plots (bottom row of Figure~\ref{fig:skew-particle-trace}) show that MAD-Path is the only method that jumps frequently among all modes. APT-DEO also traverses all modes but spends more time in the top-right corner, toward which every local mode is skewed. This asymmetry likely induces a weight distortion analogous to the one discussed in Section~\ref{sec:ideal:transition}, accounting for APT-DEO's biased weight estimates.
\end{itemize}

\begin{figure}[htbp]
    \centering
    \includegraphics[width=\linewidth]{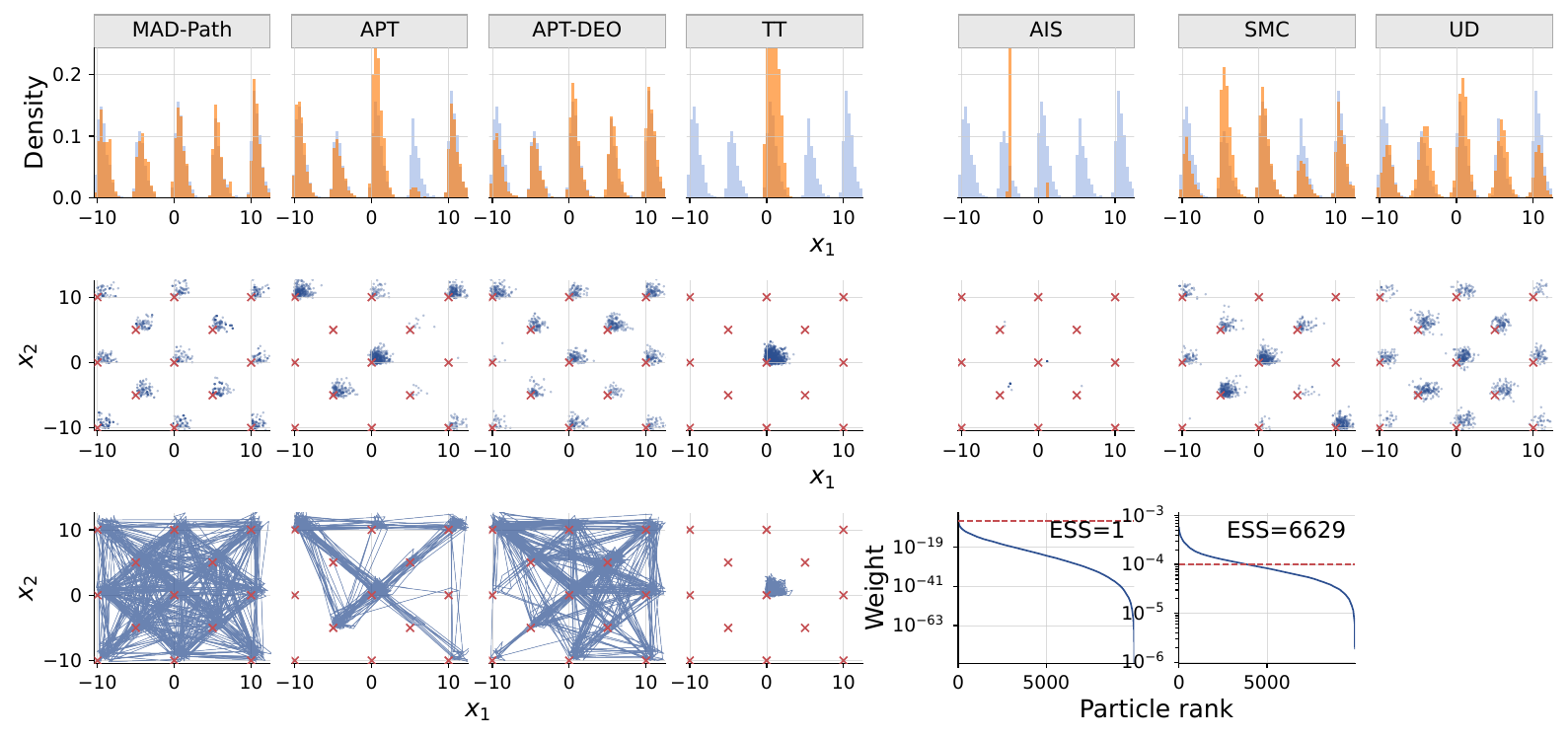}
    \caption{Results for the skew-normal mixture example. 
    First row: marginal distribution of the first coordinate (orange) against the ground truth (blue). 
    Second row: scatter plots of the first two coordinates, with the mode locations $\bmu_k$ marked in red.
    Third row: two-dimensional trace plots of the first two coordinates for the four MCMC samplers, and sorted-weight curves for AIS and SMC.}
    \label{fig:skew-particle-trace}
\end{figure}

\begin{table}[htbp]
    \centering
    \begin{tabular}{l
        S[table-format=2.4]
        S[table-format=2.4]
        S[table-format=2.4]
        S[table-format=2.4]
        S[table-format=1.4]
        S[table-format=1.4]
        S[table-format=1.4]}
        \toprule
         & {MAD-Path} & {APT} & {APT-DEO} & {TT} & {AIS} & {SMC} & {UD} \\
        \midrule
        \multirow{2}{*}{TV}
         & \textbf{0.0777} & 0.3423 & 0.1159 & 0.9231 & 0.7932 & 0.5241 & 0.2136 \\
         & {(0.0361)} & {(0.0641)} & {(0.0231)} & {(0.0000)} & {(0.0978)} & {(0.0978)} & {(0.0181)} \\
        \addlinespace
        \multirow{2}{*}{Max/Min}
         & \textbf{1.8819} & {--} & 2.8080 & {--} & {--} & {--} & 7.8484 \\
         & {(1.4190)} & {} & {(0.7594)} & {} & {} & {} & {(1.8446)} \\
        \bottomrule
    \end{tabular}
    \caption{Mode-weight estimation accuracy for the mixture of skew-normals example, over 100 independent runs (standard deviations in parentheses). TV is the total variation distance between the estimated and true (uniform) mode weights; Max/Min is the ratio of the largest to the smallest estimated weight, which equals 1 under perfect estimation. A dash (--) indicates that the method failed to recover all 13 modes, in which case the ratio is undefined. }
    \label{tab:skew}
\end{table}

\subsection{Bayesian Gaussian mixture model}
\label{sec:mixture:model}

Mixture models are canonical examples of multimodal posteriors, owing to label switching: any permutation of the component labels leaves the likelihood unchanged, producing modes of equal mass that may be well-separated in parameter space and are often inaccessible to local MCMC. As \citet{Jasra2005Markov} note, a sampler must adequately explore all such labelings to be regarded as convergent.

We consider a Bayesian Gaussian mixture model with $K = 2$ components on standardized observations $\bfy_1,\dots,\bfy_n \in \mathbb{R}^m$:
\begin{align*}
    \bfy_i \mid z_i,\bmu,\Sigma &\;\sim\; \N(\bmu_{z_i},\, \Sigma_{z_i}), \\
    z_i \mid \pi &\;\sim\; \mathrm{Cat}(\pi),
\end{align*}
with conjugate priors $\bmu_k \sim \mathcal{N}_2(\bfm_0,\, V_0)$, $\Sigma_k^{-1} \sim \mathrm{Wishart}(\nu_0,\, W_0)$, and $\pi \sim \mathrm{Dirichlet}(\alpha_0\mathbf{1}_K)$; hyperparameter settings are given in Appendix~\ref{app:experiments:details}. We use the Old Faithful geyser dataset \cite{azzalini1990look}, comprising $n = 272$ observations of $m = 2$ variables, eruption duration and waiting time to the next eruption (both in minutes). This is a standard benchmark in density estimation and clustering, notable for its clear bimodal structure. Reparametrizing all parameters into unconstrained space yields a parameter vector in $\R^{11}$; see Appendix~\ref{app:experiments:details}.

We evaluate the samplers by their accuracy in estimating the weights of the two modes, which are equal by symmetry. We run MAD-Path for 10,000 iterations and, for comparison, APT-DEO for 12,000, discarding the first 2,000 as burn-in. We use the same initialization as described in Section~\ref{sec:experiment-skew}. The experiment is repeated 100 times independently. 
The average acceptance rate of MAD-Path is 0.607, and the average swap acceptance rate of APT-DEO is 0.347. 

The left panel of Figure~\ref{fig:bayesian-gmm} shows the empirical distribution of the estimated weight of mode 1 across 100 runs. For MAD-Path, the estimates concentrate tightly around the true value $0.5$, indicating that the chain explores both modes and mixes well across them. The APT-DEO estimates, by contrast, vary widely across runs, a sign that the chains have not fully mixed. The table in the right panel reports the fraction of runs in which each sampler recovers both modes, together with the imbalance index Max/Min---the ratio of the largest to smallest estimated mode weight, computed over successful runs. MAD-Path recovers both modes in all 100 runs, with a Max/Min ratio close to $1$. APT-DEO fails to recover both modes in 6\% of runs, and its Max/Min ratio is substantially larger than $1$, confirming that the recovered weights are highly imbalanced even when both modes are found.

\begin{figure}[htbp]
    \centering
    \hspace{1em}
    \begin{minipage}{.25\linewidth}
    \includegraphics[width=\linewidth]{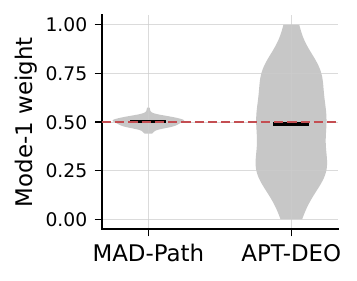}
    \end{minipage}
    \hspace{.1\linewidth}
    \begin{minipage}{.6\linewidth}
        \begin{tabular}{l
        S[table-format=2.4]
        S[table-format=2.4]}
        \toprule
         & {MAD-Path} & {APT-DEO} \\
        \midrule
        {Success rate}
         & \textbf{100\%} & {94\%} \\
        \addlinespace
        \multirow{2}{*}{Max/Min}
         & \textbf{1.0833}  & 15.4167  \\
         & {(0.0655)} & {(79.7315)} \\
        \bottomrule
        \end{tabular}
    \end{minipage}
    \caption{Results for the Bayesian Gaussian mixture model, comparing MAD-Path and APT-DEO over 100 independent runs. Left: empirical distribution of the estimated mode-1 weight, with the true value 0.5 marked by the red dashed line. Right: success rate (proportion of runs identifying both modes), Max/Min (ratio of the largest to smallest estimated mode weight for successful runs, equal to 1 under perfect estimation), 
    with standard deviations in parentheses.}
    \label{fig:bayesian-gmm}
\end{figure}

\subsection{Effects of diffusion horizon and discretization steps}
\label{sec:fetal:deaths}
We study the effects of the horizon $T$ and the number of discretization steps $N$ on MAD-Path using the mixture model of beta-binomial and binomial distributions from \cite{tjelmeland2001mode}. The dataset consists of observations on fetal deaths in litters of mice \cite{brooks1997finite}. The conditional distribution of death given the number of implants or fetuses is defined by
\begin{align*}
    p(x\mid n) = \gamma \left[ \binom{n}{x} \prod_{r=0}^{x-1} \frac{\mu + r\theta}{1 + r\theta} \prod_{r=0}^{n-x-1} \frac{1-\mu + r\theta}{1 + (r+x)\theta} \right] + (1 - \gamma) \left[  \binom{n}{x} \nu^x (1-\nu)^{n-x} \right], 
\end{align*}
where $x$ is the number of deaths of mice and $n$ the number of implants or fetuses, the mixing weight $\gamma\in[0,1]$, and where $\mu,\nu \in [0,1]$, $\theta\geq0$. Further details of the model are provided in Appendix~\ref{app:experiments:details}.

We evaluate MAD-Path across a range of values of $T$ and $N$, measuring the acceptance rate and the inter-mode crossings per 1,000 steps. Note that during training, we set $T=7$ and $N=100$. The experiments are repeated 100 times independently, and the median and interquartile range of the acceptance rate and inter-mode crossings are reported in Figure~\ref{fig:fetal-deaths}. 

The left panel of Figure~\ref{fig:fetal-deaths} shows that the acceptance rate increases with $N$, as the discretization error decreases. Conversely, it decreases with $T$, since a larger $T$ incurs a larger score error $\calE_{\mathrm{score}}$ defined in Equation~\eqref{eq:integrated-score-error}. With $T=7$, the acceptance rate plateaus around $0.35$, at which point score estimation error dominates the discretization error and larger $N$ does not improve the acceptance rate further.

The right panel of Figure~\ref{fig:fetal-deaths} shows that the inter-mode crossing rate increases with both $T$ and $N$. In particular, for $T\leq 4$, there is almost no mode-jumping. This corroborates the analysis in Section~\ref{sec:ideal:transition}, which shows that mixing across modes requires a sufficiently long diffusion horizon for the modes to merge by the end of the forward process. 
Trace plots for various combinations of $T$ and $N$ are provided in Figure~\ref{fig:fetal-deaths-traces} in Appendix~\ref{app:experiments:details}.

\begin{figure}[htbp]
    \centering
    \includegraphics[width=\linewidth]{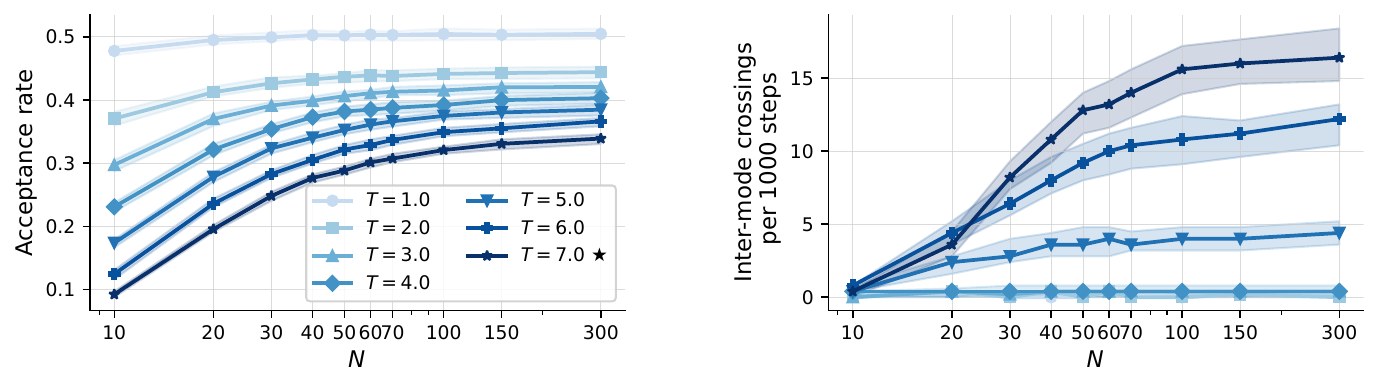}
    \caption{MAD-Path acceptance rate (left) and inter-mode crossing rate (right) with varying horizon $T$ and number of steps $N$ on the fetal deaths in litters example. Lines show the median over 100 replications; shaded regions indicate the interquartile range (25th--75th percentiles). 
    } 
    \label{fig:fetal-deaths}
\end{figure}

\subsection{Sensor network localization}
\label{sec:sensor:network}

We compare MAD-Path against APT-DEO during the initial burn-in phase on a sensor network localization problem studied, for example, in~\citep{ihler2004nonparametric, ahn2013distributed, lan2014wormhole, pompe2020framework}. Eleven sensors are placed in $\R^2$, three of which have known locations. The goal is to infer the remaining eight unknown locations $\bfx_1,\dots,\bfx_8 \in \R^2$ from noisy inter-sensor distance measurements. Let $w_{ij}\in\{0,1\}$ indicate whether the distance between sensors $i$ and $j$ is observed; it follows a Bernoulli distribution with success probability $\exp\!\left(-\frac{\|\bfx_i - \bfx_j\|^2}{2 \times 0.25^2}\right)$. When a measurement is obtained, the observed distance $y_{ij}$ follows $\N\!\left(\|\bfx_i - \bfx_j\|,\, 0.01^2\right)$. Following~\cite{ahn2013distributed, lan2014wormhole, pompe2020framework}, we place an improper uniform prior on each sensor location $\bfx_i$, $1\leq i\leq 8$. The resulting posterior is
\begin{align*}
    p (\bfx_1,\ldots, \bfx_8 \mid y,w) \propto \prod_{\substack{1\leq i\leq 8 \\ i< j \leq 11 }} f_{ij} (\bfx_i, \bfx_j \mid y_{ij}, w_{ij}),
\end{align*}
where
\begin{align*}
    f_{ij} (\bfx_i, \bfx_j \mid y_{ij}, w_{ij}) = \begin{cases}
        \exp \left(-\frac{\| \bfx_i - \bfx_j \|^2}{2 \times 0.25^2}\right)\, \exp\left(-\frac{ (y_{ij} - \| \bfx_i - \bfx_j \|)^2 }{2 \times 0.01^2}\right), & \text{if}\; w_{ij} = 1\\
        1 - \exp\left( -  \frac{\| \bfx_i - \bfx_j \|^2}{2 \times 0.25^2}\right), & \text{otherwise}.
    \end{cases} 
\end{align*}

Sampling from this posterior presents two sources of difficulty. First, given the location of sensor $j$, each observed distance $y_{ij}$ concentrates the marginal density of sensor $i$ near a circle of radius $y_{ij}$ centered at $\bfx_j$; when a sensor is constrained by two or more such distances, these marginals form banana-shaped ridges. Second, the sparsity of observations induces near non-identifiability and, consequently, multimodality. Figure~\ref{fig:sensor-network} in Appendix~\ref{app:experiments:details} visualizes the posterior.

For MAD-Path, we use the natural initialization from the learned reverse diffusion path. For APT-DEO, we consider two initializations, $0.5 \cdot \boldsymbol{1}_d$ and $\boldsymbol{1}_d$. The average acceptance rate of MAD-Path is $0.254$, and the average swap acceptance rates of the two APT-DEO chains are $0.271$ and $0.461$, respectively. Figure~\ref{fig:sensor-network-traces} shows the first 500 iterations of each chain under these initializations. APT-DEO requires hundreds of iterations to locate the posterior mass during its transient phase. MAD-Path exhibits essentially no such transient phase: the mode-finding is instead absorbed into the score-learning stage, so the sampler starts from a distribution that already covers the modes.

\begin{figure}[htbp]
    \centering
    \includegraphics[width=\linewidth]{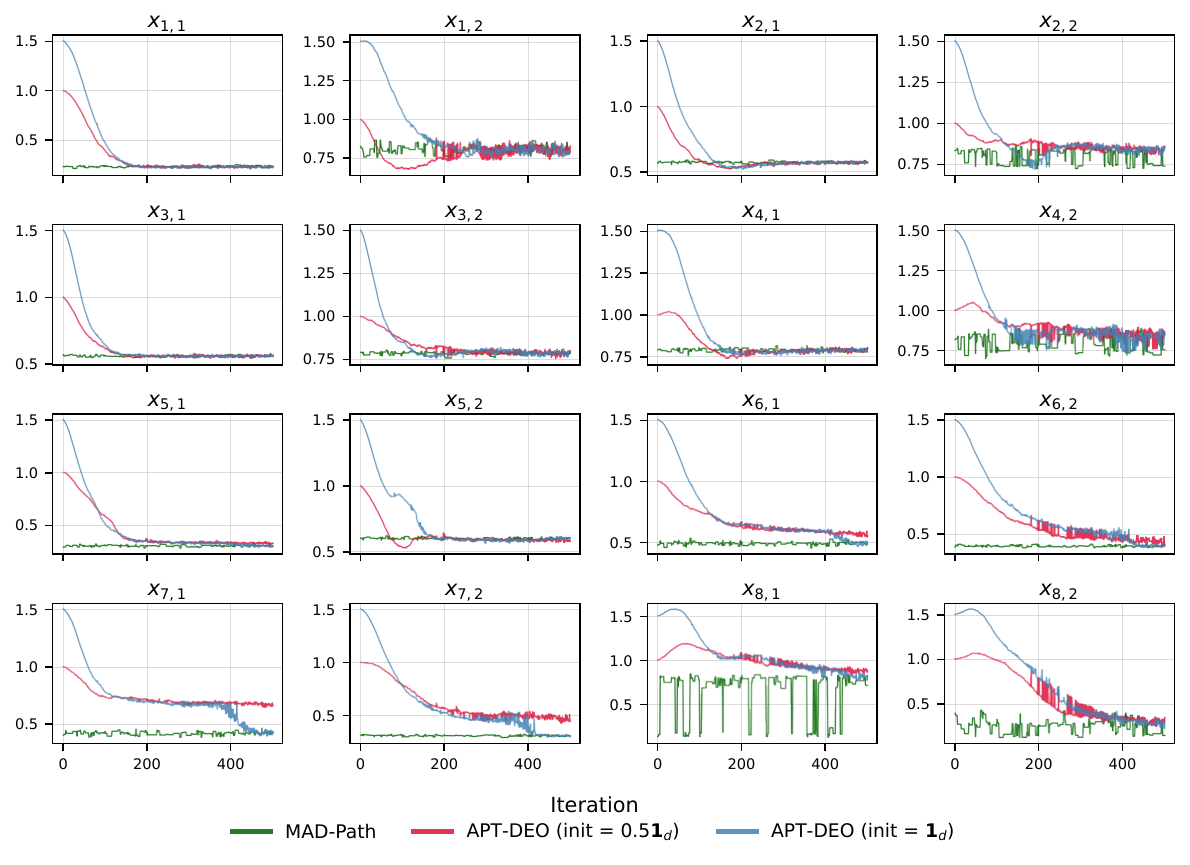}
    \caption{Burn-in phase (first 500 iterations) of MAD-Path (green), APT-DEO with initial state at $0.5 \boldsymbol{1}_d$ (red), and APT-DEO with initial state at $\boldsymbol{1}_d$ (blue).}
    \label{fig:sensor-network-traces}
\end{figure}

\subsection{Seemingly unrelated regression}
\label{sec:sur}

A well-known multimodal posterior arises from the seemingly unrelated regression (SUR) model~\citep{zellner1962efficient, srivastava1987seemingly}. The model comprises $M$ regression equations, each with its own response vector and design matrix but sharing a common error covariance:
\begin{align*}
    \bfy_m = X_m \boldsymbol{\beta}_m + \boldsymbol{\varepsilon}_m, \quad m = 1, \ldots, M,
\end{align*}
where $\bfy_m \in \mathbb{R}^N$ is the response for equation $m$, $X_m \in \mathbb{R}^{N \times k_m}$ is the fixed design matrix, and $\boldsymbol{\beta}_m \in \mathbb{R}^{k_m}$ is the coefficient vector. The errors are jointly distributed as
\begin{align*}
    \boldsymbol{\varepsilon} = (\boldsymbol{\varepsilon}_1, \ldots, \boldsymbol{\varepsilon}_M) \sim \mathcal{N}(0,\, \Sigma \otimes I_N),
\end{align*}
with unknown covariance $\Sigma \in \mathbb{R}^{M \times M}$. We denote the full coefficient vector by $\boldsymbol{\beta} = (\boldsymbol{\beta}_1, \ldots, \boldsymbol{\beta}_M) \in \mathbb{R}^{\sum_m k_m}$. Profiling out the covariance by its MLE yields the profile log-likelihood
\begin{equation}
    \label{eq:sur}
    \ell(\boldsymbol{\beta}) = -\frac{NM}{2}\log(2\pi) - \frac{N}{2}\log\det\left(G(\boldsymbol{\beta})\right) - \frac{NM}{2},
\end{equation}
where $G(\boldsymbol{\beta}) =\frac1N R(\boldsymbol{\beta})^\top R(\boldsymbol{\beta})$ and $R(\boldsymbol{\beta}) = (\bfr_1 \, \cdots \, \bfr_M) \in \mathbb{R}^{N\times M}$ collects the residual vectors $\bfr_m(\boldsymbol{\beta}) = \bfy_m - X_m  \boldsymbol{\beta}_m \in \R^{N}$. Such profile likelihoods can exhibit multimodality in small samples~\cite{drton2004multimodality}. We take $M=5$ equations with scalar coefficients ($k_m = 1$ for all $m$) and $N=10$ observations per equation. This small-$N$ regime induces strong finite-sample multimodality: the profile likelihood exhibits three distinct modes.

MAD-Path is run for 10,000 iterations, achieving an average acceptance rate of $0.327$. APT-DEO is run for 12,000 iterations with the first 2,000 discarded as burn-in; it is initialized with a warm start near one mode and repeated across two random seeds, attaining average swap acceptance rates of $0.467$ and $0.470$, respectively.
Figure~\ref{fig:sur-hist} shows the marginal sample distributions from MAD-Path and APT-DEO, with a reference computed by the No-U-Turn Sampler (NUTS)~\citep{hoffman2014nuts} over 50{,}000 iterations (dashed line). 
MAD-Path recovers sharp, well-separated peaks that align closely with the reference mode locations. In contrast, APT-DEO is unstable across random seeds: the relative heights of the three modes vary noticeably between runs, indicating that the chain has not fully mixed and the estimated mode weights have not settled. Under a cold-start initialization, the transient phase is even longer, further delaying convergence; see Figure~\ref{fig:sur-trace} in Appendix~\ref{app:experiments:details}.

\begin{figure}[htbp]
    \centering
    \includegraphics[width=\linewidth]{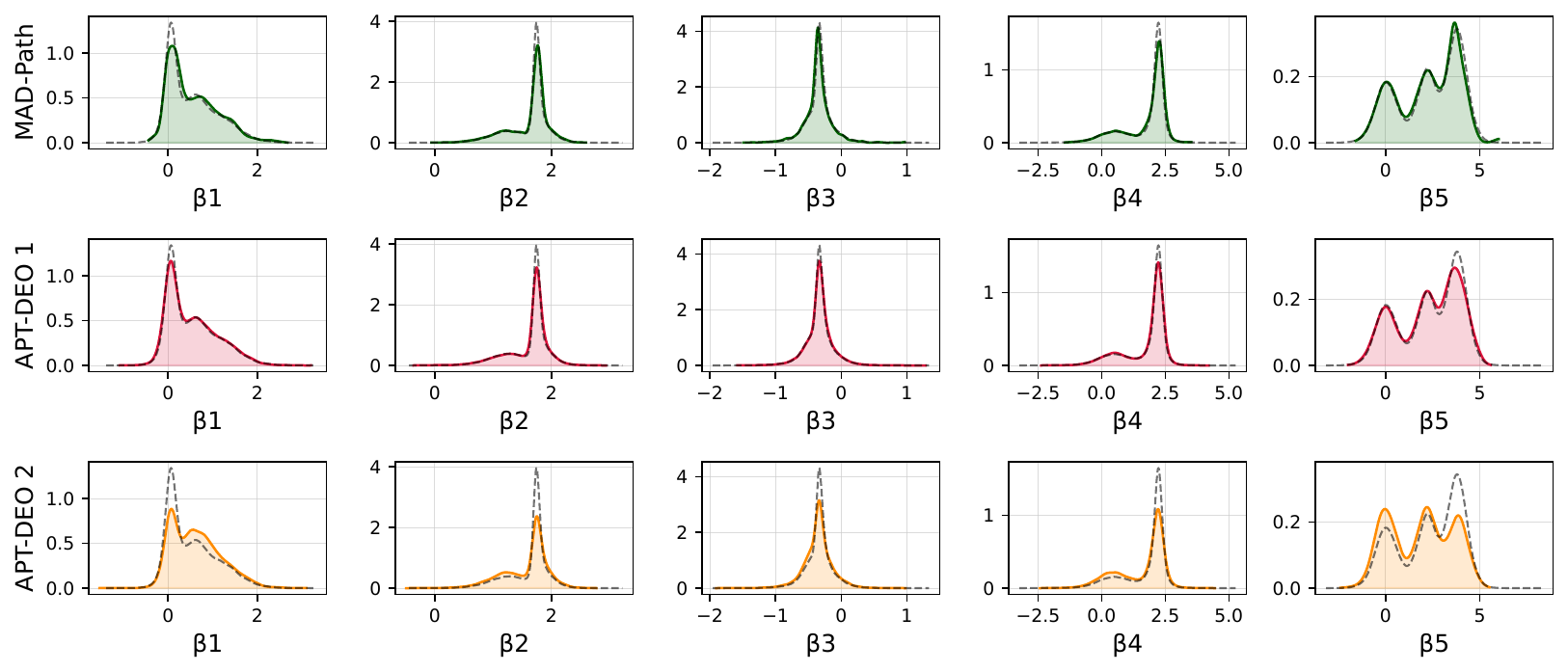}
    \caption{Marginal sample distributions for the SUR example. The top row shows MAD-Path; the bottom two rows show APT-DEO under two different random seeds. Black dashed lines indicate the reference distribution obtained via NUTS.}
    \label{fig:sur-hist}
\end{figure}

\section{Discussion}
\label{sec:discussion}

We introduced MAD-Path, a Metropolis--Hastings sampler that builds nonlocal proposals from a forward--backward diffusion path. The sampler leaves the target invariant even under an approximate score function and discretization error. We showed that the ideal diffusion-path transition has favorable mixing properties, and we characterized how the acceptance probability of MAD-Path degrades under score and discretization error. Finally, we discussed practical considerations and demonstrated the method on a range of challenging multimodal targets.

Several avenues remain for future work. First, practical guidelines are needed for jointly tuning the diffusion horizon $T$ and the number of discretization steps $N$, balancing mode overlap and acceptance probability against computational cost. Second, existing training objectives minimize a one-sided path-space KL divergence~\eqref{equ:reverse:path:kl}, whereas the acceptance probability of MAD-Path is governed by a symmetrized path-space KL (Theorem~\ref{thm:score-error-mean}). Developing tractable objectives that more directly target this symmetrized divergence may therefore improve sampling efficiency. Finally, trajectories generated during sampling could be used to refine the score online, suggesting adaptive variants of MAD-Path.

On the theoretical side, it would be interesting to establish the mixing time of MAD-Path. The present analysis separates the ideal diffusion-path transition from the effects of score and discretization error. A promising route is to view MAD-Path as a hybrid data-augmentation chain on path space. The comparison theory of Markov chains \cite{andrieu2022comparison,qin2025spectral} may then yield spectral-gap bounds for MAD-Path by combining the gap of the ideal diffusion-path kernel (Section~\ref{sec:ideal:transition}) with a conditional mixing bound for the path-space Metropolis update.

\subsection*{Acknowledgments}

We thank Yifan Chen, Marylou Gabrie, Louis Grenioux, Guanyang Wang, Esmeralda Whitammer, Luhuan Wu for helpful discussions on diffusion-based samplers, and Changwoo Lee for discussions on applications of this work. HC acknowledges support from the Danish Data Science Academy (DDSA) Visit Grant (2026-5978). JY acknowledges support from the Independent Research Fund Denmark (DFF) through the Sapere Aude Starting Grant (5251-00032B).

\appendix

\section{Proofs for Section~\ref{sec:ideal:transition}}

\subsection{Proof of Lemma~\ref{lem:gap:maximal:corr}}
\label{prf:gap:maximal:corr}
\begin{proof}[Proof of Lemma~\ref{lem:gap:maximal:corr}]
    Let $\calT$ and $\calT^*$ denote the operators associated with the forward transition $\vec{P}$ and backward transition $\leftvec{P}$. For any $f\in L^2(p_0)$ and $g\in L^2(p_T)$, we have
    \begin{align*}
        \langle f, \calT g\rangle_{L^2(p_0)} = \langle \calT^* f, g\rangle_{L^2(p_T)}
    \end{align*}
    and $P=\calT \calT^*$. Thus, for any $f\in L^2(p_0)$ and $\EE[p_0]{f}=0$,
    \begin{align*}
        \langle f, Pf\rangle_{L^2(p_0)} = \|\calT^* f\|_{L^2(p_T)}^2=\Var{\EE{f(X_0)\mid X_T}}.
    \end{align*}
    Therefore, we have
    \begin{align*}
        \mathrm{Gap}(P)=1-\sup_{\substack{f\in L^2(p_0)\\ \EE[p_0]{f}=0, f\neq 0} } \frac{\langle f, Pf\rangle_{L^2(p_0)}}{\|f\|_{L^2(p_0)}^2} = 1-\sup_{\substack{f\in L^2(p_0)\\ \EE[p_0]{f}=0,f\neq 0 }} \frac{\Var{\EE{f(X_0)\mid X_T}}}{\Var{f(X_0)}}.
    \end{align*}
    The result follows from the definition of maximal correlation and the fact that
    \begin{align}\label{eq:maximal-corr-conditional-variance}
    \sup_{\substack{f\in L^2(p_0)\\ g\in L^2(p_T)}} \mathrm{Corr}(f(X_0), g(X_T))^2 = \sup_{f\in L^2(p_0)} \frac{\Var{\EE{f(X_0)\mid X_T } } }{\Var{f(X_0)}}.
    \end{align}
\end{proof}

\subsection{Proof of Theorem~\ref{thm:gap:poincare}}
\label{prf:gap:poincare}
\begin{proof}[Proof of Theorem~\ref{thm:gap:poincare}]
Let $X\sim p$ and $\bfxi\sim \N(0,I_d)$ be independent.
Let $Z = \sqrt{\alpha} X + \sqrt{1-\alpha} \bfxi$ and let $p_Z$ denote the distribution of $Z$. 
Let $\calT:L^2(p_Z)\to L^2(p)$ be the forward conditional expectation operator
\[
    (\calT g)(\bfx)
    =
    \E[g(Z)\mid X=\bfx]
    =
    \int
    g(\sqrt{\alpha}\bfx+\sqrt{1-\alpha}\bfxi)
    \varphi(\bfxi) \rd \bfxi.
\]
Since the maximal correlation can be equivalently expressed as
\[
    \sup_{g\in L^2(p_Z)} \frac{\Var{\EE{g(Z)\mid X } } }{\Var{g(Z)}},
\]
it suffices to prove that, for every $g\in L^2(p_Z)$ with $\EE[p_Z]{g}=0$,
\[
    \Var[p]{\calT g}
    \le
    \frac{\alpha C_{\mathrm P}}
    {1-\alpha+\alpha C_{\mathrm P}}
    \Var[p_Z]{g}.
\]
Applying integration by parts, 
\[
    \nabla (\calT g)(\bfx)= \sqrt{\frac{\alpha}{1-\alpha}}\, \EE{\bfxi g(Z) \mid X=\bfx}
    = \sqrt{\frac{\alpha}{1-\alpha}}\, \EE{\bfxi (g(Z) - (\calT g)(\bfx)) \mid X=\bfx}.
\]
For any unit vector $\bfu\in\R^d$,
\[
    \bfu\tran \nabla (\calT g)(\bfx)
    = \sqrt{\frac{\alpha}{1-\alpha}}\, \EE{(\bfu\tran \bfxi) (g(Z) - (\calT g)(\bfx)) \mid X=\bfx}.
\]
Applying Cauchy--Schwarz inequality in the conditional law,
\begin{align*}
    \big|\bfu\tran \nabla (\calT g)(\bfx)\big|^2
    \leq \frac{\alpha}{1-\alpha} \EE{(\bfu\tran \bfxi)^2 \mid X=\bfx} \cdot \EE{(g(Z) - (\calT g)(\bfx))^2 \mid X=\bfx}.
\end{align*}
Therefore,
\begin{align*}
\|\nabla (\calT g)(\bfx)\|_2^2 \leq \frac{\alpha}{1-\alpha} \EE{g(Z)^2 - (\calT g)(\bfx)^2 \mid X=\bfx }.
\end{align*}
By the Poincar\'e inequality,
\begin{align*}
    \Var[p]{\calT g} &\leq C_{\mathrm P} \int \|\nabla (\calT g)(\bfx)\|_2^2 p(\rd \bfx) \leq \frac{\alpha C_{\mathrm P}}{1-\alpha} \EE[p]{ \EE{g(Z)^2 - (\calT g)(\bfx)^2 \mid X=\bfx } }\\
    &= \frac{\alpha C_{\mathrm P}}{1-\alpha} \Bigl(\EE{g(Z)^2 } - \EE{(\calT g)(X) ^2 }  \Bigr) \\
    &= \frac{\alpha C_{\mathrm P}}{1-\alpha} \Bigl(\Var[p_Z]{g } - \Var[p]{\calT g }  \Bigr) .
\end{align*}
Rearranging the above inequality gives the desired result.

When $p=\N(\mu,\Sigma)$ and $Z\mid X=\N(\sqrt\alpha X, (1-\alpha)I_d)$, $X$ and $Z$ are jointly Gaussian, and their maximal correlation is equal to the maximum linear correlation~\citep{dembo2001remarks,bryc2005maximum}, which is given by
\begin{align*}
    \rho_{\max}(X,Z)^2=\frac{\alpha \lambda_{\max}(\Sigma)}{1-\alpha+ \alpha \lambda_{\max}(\Sigma) }.
\end{align*}

The upper bound on $\rho_{\max}(X,Z)$ implies a lower bound on $\mathrm{Gap}(P)$ following Lemma~\ref{lem:gap:maximal:corr}. The proof is complete.
\end{proof}

\subsection{Proof of Proposition~\ref{prop:lifted:corr:decomposition}}
\label{prf:lifted:corr:decomposition}

\begin{proof}[Proof of Proposition~\ref{prop:lifted:corr:decomposition}]
Let $X\sim p$ and $\bfxi\sim\N(0,I_d)$ be independent, and let
$Z=\sqrt\alpha X+\sqrt{1-\alpha}\bfxi$.
Write $\vec{P}(\bfx,\rd\bfz)=\vec{p}(\bfx,\bfz)\,\rd\bfz$.
For each label $c\in\calC$, define the noised component density
\begin{align}\label{eq:define-q_c}
    q_c(\bfz)
    :=\int_{\mathbb R^d}
    p_c(\bfx)\vec{p}(\bfx,\bfz)\,\rd\bfx,
\end{align}
and let $q(\bfz)=\sum_{c\in\calC}w_cq_c(\bfz)$ denote the marginal
density of $Z$.

Let $g\in L^2(q)$ be centered with $\Var[q]{g}=1$.
By the law of total variance and the equivalent definition of maximal 
correlation in Equation~\eqref{eq:maximal-corr-conditional-variance},
\begin{align*}
    \EE{\Var{g(Z)\mid C}}
    &=
    1-\Var{\EE{g(Z)\mid C}} \geq
    1-\rho_{\mathrm{lab}}^2.
\end{align*}
Similarly, for each $c\in\calC$, the definition of $\rho_c^2$ gives
\[
    \Var{\EE{g(Z)\mid X,C}\mid C=c}
    \leq
    \rho_c^2\Var{g(Z)\mid C=c}.
\]
Applying the law of total variance under the conditional law given
$C=c$, we obtain
\begin{align*}
    \EE{\Var{g(Z)\mid X,C}\mid C=c}
    &=
    \Var{g(Z)\mid C=c}
    -
    \Var{\EE{g(Z)\mid X,C}\mid C=c} \\
    &\geq
    (1-\rho_c^2)\Var{g(Z)\mid C=c}.
\end{align*}
Averaging over $C$ therefore yields
\begin{align*}
    \EE{\Var{g(Z)\mid X,C}}
    &\geq
    \sum_{c\in\calC}
    w_c(1-\rho_c^2)\Var{g(Z)\mid C=c} \\
    &\geq
    \left(1-\max_{c\in\calC}\rho_c^2\right)
    \EE{\Var{g(Z)\mid C}} \\
    &\geq
    \left(1-\max_{c\in\calC}\rho_c^2\right)
    \bigl(1-\rho_{\mathrm{lab}}^2\bigr).
\end{align*}
Since
\[
    \Var{\EE{g(Z)\mid X,C}}
    =
    1-\EE{\Var{g(Z)\mid X,C}},
\]
taking the supremum over all centered, unit-variance $g\in L^2(q)$
gives
\begin{align*}
    \tilde\rho_{\max}^2
    &=
    \sup_{\substack{g\in L^2(q)\\ \EE[q]{g}=0,\ \Var[q]{g}=1}}
    \Var{\EE{g(Z)\mid X,C}} \\
    &\leq
    1-
    \bigl(1-\rho_{\mathrm{lab}}^2\bigr)
    \left(1-\max_{c\in\calC}\rho_c^2\right) \\
    &=
    \rho_{\mathrm{lab}}^2
    +
    \bigl(1-\rho_{\mathrm{lab}}^2\bigr)
    \max_{c\in\calC}\rho_c^2.
\end{align*}
The lower bound on the spectral gap follows from
$\mathrm{Gap}(P)= 1-\tilde\rho_{\max}^2$.
\end{proof}

\subsection{Proof of Theorem~\ref{thm:mixture:poincare:gap:bound}}
\label{prf:mixture:poincare:gap:bound}

\begin{proof}[Proof of Theorem~\ref{thm:mixture:poincare:gap:bound}]
Since each $p_c$ satisfies a Poincar\'e inequality, Theorem~\ref{thm:gap:poincare} implies
\begin{align*}
    \mathrm{Gap}(P_c) \geq \frac{1-\alpha}{1-\alpha + \alpha C_{\mathrm P}}.
\end{align*}
Following Proposition~\ref{prop:lifted:corr:decomposition}, it suffices to lower bound the label transition gap $\mathrm{Gap}(P_{\mathrm{lab}})=1-\rho_{\mathrm{lab}}^2$. 

For distinct $c,c'\in\calC$, define the pairwise overlap
\begin{align}\label{eq:define-lambda}
\lambda_{c,c'}^{(w)}
:=
\int
\frac{q_c(\bfz)q_{c'}(\bfz)}
     {q(\bfz)}
\,\rd\bfz,
\qquad
q(\bfz)=\sum_{k\in\calC}w_kq_k(\bfz),
\end{align}
where $q_c$ is the density of the noised component corresponding to class $c$, as defined in Equation~\eqref{eq:define-q_c}.
Lemma~\ref{lem:lambda:expression} below gives
\[
\mathrm{Gap}(P_{\mathrm{lab}})
\geq
\min_{\substack{c,c'\in\calC\\c\neq c'}}
\lambda_{c,c'}^{(w)}.
\]
By Lemma~\ref{lem:overlap:wasserstein},
\[
\lambda_{c,c'}^{(w)}
\geq
\exp\left\{
-\frac{\alpha}{4(1-\alpha)}
W_2^2(p_c,p_{c'})
\right\}.
\]
Combining this inequality with Proposition~\ref{prop:lifted:corr:decomposition} and the bounds on $\mathrm{Gap}(P_c)$ yields the desired result.

\end{proof}

\begin{lem}[Label transition gap via overlap]
\label{lem:lambda:expression}
Recall $\lambda_{c,c'}^{(w)}$ is defined as in Equation~\eqref{eq:define-lambda}.
We have
\[
\mathrm{Gap}(P_{\mathrm{lab}})
\geq
\min_{\substack{c,c'\in\calC\\c\neq c'}}
\lambda_{c,c'}^{(w)}.
\]

\end{lem}
\begin{proof}[Proof of Lemma~\ref{lem:lambda:expression}]
    Given $C=c$, the variable $Z$ has density $q_c$, while
\[
\PP{C'=c'\mid Z=\bfz}
=
\frac{w_{c'}q_{c'}(\bfz)}{q(\bfz)}.
\]
Hence the label-transition kernel is
\[
P_{\mathrm{lab}}(c,c')
=
w_{c'}
\int
\frac{q_c(\bfz)q_{c'}(\bfz)}
     {q(\bfz)}
\,\rd\bfz
=
w_{c'}\lambda_{c,c'}^{(w)}.
\]
In particular,
\[
w_cP_{\mathrm{lab}}(c,c')
=
w_cw_{c'}\lambda_{c,c'}^{(w)}
=
w_{c'}P_{\mathrm{lab}}(c',c),
\]
so $P_{\mathrm{lab}}$ is reversible with respect to the mixture weights $\{w_c\}_{c\in\calC}$.

For any $h:\calC\to\mathbb R$, its Dirichlet form satisfies
\begin{align*}
\mathcal E_{\mathrm{lab}}(h,h)
&=
\frac12
\sum_{c,c'\in\calC}
w_cP_{\mathrm{lab}}(c,c')
\bigl(h(c)-h(c')\bigr)^2 \\
&=
\frac12
\sum_{c,c'\in\calC}
w_cw_{c'}\lambda_{c,c'}^{(w)}
\bigl(h(c)-h(c')\bigr)^2 \\
&\geq
\frac12
\min_{\substack{c,c'\in\calC\\c\neq c'}}
\lambda_{c,c'}^{(w)}
\sum_{c,c'\in\calC}
w_cw_{c'}
\bigl(h(c)-h(c')\bigr)^2 \\
&=
\min_{\substack{c,c'\in\calC\\c\neq c'}}
\lambda_{c,c'}^{(w)}
\Var{h(C)}.
\end{align*}
Taking the infimum over all nonconstant $h$ proves the result.

\end{proof}

\begin{lem}[Overlap lower bound via Wasserstein distance]
\label{lem:overlap:wasserstein}
For every distinct $c,c'\in\calC$,
\[
\lambda_{c,c'}^{(w)}
\geq
\exp\left\{
-\frac{\alpha}{4(1-\alpha)}
W_2^2(p_c,p_{c'})
\right\}.
\]
\end{lem}
\begin{proof}[Proof of Lemma~\ref{lem:overlap:wasserstein}]
By the Cauchy--Schwarz inequality,
\begin{align}
\lambda_{c,c'}^{(w)}
&=
\int
\frac{q_c(\bfz)q_{c'}(\bfz)}
     {q(\bfz)}
\,\rd\bfz
\cdot
\int q(\bfz)\,\rd\bfz \geq
\left(
\int\sqrt{q_c(\bfz)q_{c'}(\bfz)}\,\rd\bfz
\right)^2
=:I_{\alpha;c,c'}^2.
\label{equ:overlap:lower:bound}
\end{align}
Let $\Gamma$ be any coupling of $p_c$ and $p_{c'}$, i.e. $\Gamma$ is a probability measure on $\mathbb R^d\times \mathbb R^d$ with marginals $p_c$ and $p_{c'}$. Then we can write
\[
q_c(\bfz)
=
\int
\varphi\bigl(
\bfz;\sqrt{\alpha}\bfx,(1-\alpha)I_d
\bigr)
\,\Gamma(\rd\bfx,\rd\bfy),\quad
q_{c'}(\bfz)
=
\int
\varphi\bigl(
\bfz;\sqrt{\alpha}\bfy,(1-\alpha)I_d
\bigr)
\,\Gamma(\rd\bfx,\rd\bfy).
\]
Applying the Cauchy--Schwarz inequality with respect to $\Gamma$ gives
\begin{align*}
    I_{\alpha;c,c'}&= \int \sqrt{q_c(\bfz)q_{c'}(\bfz)}\,\rd\bfz \geq \int \sqrt{
        \varphi(\bfz; \sqrt\alpha \bfx, (1-\alpha)I_d)
        \varphi(\bfz; \sqrt\alpha \bfy, (1-\alpha)I_d)
    }\, \rd\bfz \Gamma(\rd\bfx,\rd\bfy)\\
    &=\int \left(2\pi(1-\alpha)\right)^{-d/2} \Exp{-\frac{\|\bfz-\sqrt\alpha\bfx\|^2+\|\bfz-\sqrt\alpha\bfy\|^2}{4(1-\alpha)}} \rd\bfz \Gamma(\rd\bfx,\rd\bfy) \\
    &= \int \Exp{-\frac{\alpha\|\bfx-\bfy\|^2}{8(1-\alpha)}}\Gamma(\rd\bfx,\rd\bfy).
\end{align*}
By Jensen's inequality,
\[
I_{\alpha;c,c'}
\geq
\exp\left\{
-\frac{\alpha}{8(1-\alpha)}
\int
\|\bfx-\bfy\|^2
\,\Gamma(\rd\bfx,\rd\bfy)
\right\}.
\]
Taking $\Gamma$ to be an optimal coupling of $p_c$ and $p_{c'}$ yields
\[
I_{\alpha;c,c'}
\geq
\exp\left\{
-\frac{\alpha}{8(1-\alpha)}
W_2^2(p_c,p_{c'})
\right\}.
\]
Squaring this bound and applying
Equation~\eqref{equ:overlap:lower:bound} completes the proof.

\end{proof}

\subsection{Proof of Corollary~\ref{cor:gaussian_mixture_dimension_free_gap}}
\label{prf:gaussian_mixture_dimension_free_gap}

\begin{proof}[Proof of Corollary~\ref{cor:gaussian_mixture_dimension_free_gap}]
Let $m=\sqrt\alpha b\bone_d$, $a_{\pm}^2=\alpha \sigma_{\pm}^2 + (1-\alpha)$. 
Direct calculation gives
\begin{align*}
    I_\alpha&=\int \sqrt{q_+(\bfz) q_-(\bfz) }\rd \bfz =\frac{1}{(2\pi)^{d/2} (a_+ a_-)^{d/2}} \int \exp\Bigl(-\frac{\|\bfz - m\|^2}{4a_+^2} -\frac{\|\bfz + m\|^2}{4a_-^2}\Bigr) \rd \bfz\\
    &=\Bigl(\frac{2a_+a_-}{a_+^2 + a_-^2} \Bigr)^{d/2} \exp\Bigl(-\frac{\|m\|^2}{(a_+^2 + a_-^2)}\Bigr).
\end{align*}
When $\alpha=\kappa/d$,
\[
    a_c^2
    =
    1+\frac{\kappa(\sigma_c^2-1)}{d},
    \qquad c\in\{+,-\}.
\]
Thus
\[
    a_c^2=1+O(d^{-1}),
    \qquad
    a_c=1+O(d^{-1}).
\]
Moreover, $\|m\|^2=\alpha b^2 d=\kappa b^2$.
Therefore,
\[
    \left(
        \frac{2a_+a_-}{a_+^2+a_-^2}
    \right)^d
    =
    e^{o(1)}
\]
and
\[
    \exp\left\{
        -\frac{2\|m\|^2}{a_+^2+a_-^2}
    \right\}
    =
    \exp\{-\kappa b^2+o(1)\}.
\]
Hence, by Equation~\eqref{equ:overlap:lower:bound},
\[
    \lambda_{+,-}^{(w)} \ge I_\alpha^2
    \ge
     e^{-\kappa b^2+o(1)}.
\]
On the other hand, since $\sigma_+^2>\sigma_-^2$,
\[
    1-\max_{c\in\{+,-\}}\rho_c^2
    =
    \frac{1-\alpha}{1-\alpha+\alpha\sigma_+^2}
    \longrightarrow 1
\]
as $d\to\infty$ with $\alpha=\kappa/d$.
Therefore, by Proposition~\ref{prop:lifted:corr:decomposition} and Lemma~\ref{lem:lambda:expression}
\[
    \mathrm{Gap}(P)
    \geq
    \lambda_{+,-}^{(w)}
    \left(1-\max_{c\in\{+,-\}}\rho_c^2\right)
    \geq
    I_\alpha^2
    \frac{1-\alpha}{1-\alpha+\alpha\sigma_+^2}.
\]
Since $I_\alpha^2\to e^{-\kappa b^2}$, it follows that
\[
    \liminf_{d\to\infty}\mathrm{Gap}(P)
    \geq e^{-\kappa b^2}.
\]
\end{proof}

\section{Proofs for Section~\ref{sec:analysis}}

\subsection{Proof of Theorem~\ref{thm:exact-score-discretization}}
\label{prf:exact-score-discretization}
\begin{proof}[Proof of Theorem~\ref{thm:exact-score-discretization}]
Let
\[
    b_t^\star(\bfx):=-f_t(\bfx)+g_t^2s_t(\bfx)
\]
denote the exact drift used in the backward pass. 
Denote
\begin{align*}
    A_t(\bfx):=\nabla f_t(\bfx)
    -\frac{g_t^2}{2}\nabla s_t(\bfx),
\end{align*}
which is a symmetric matrix by the assumption that $\nabla f_t$ is symmetric.
Set $h=T/N$, $t_k=kh$ and abbreviate
$f_k=f_{t_k}$,
$b_k^\star=b_{t_k}^\star$,
$g_k=g_{t_k}$, $L_k=L_{t_k}$, $A_k=A_{t_k}$.
The Euler--Maruyama paths satisfy
\begin{align*}
    \bfx_{k+1}
    &=
    \bfx_k+hf_k(\bfx_k)
    +\sqrt h\,g_k\bfep_k,
    &
    \bfep_k
    &\stackrel{\mathrm{i.i.d.}}{\sim}\mathcal N(\mathbf{0},I_d),
    \\
    \bfy_k
    &=
    \bfy_{k+1}+hb_k^\star(\bfy_{k+1})
    +\sqrt h\,g_k\bfzeta_k,
    &
    \bfzeta_k
    &\stackrel{\mathrm{i.i.d.}}{\sim}\mathcal N(\mathbf{0},I_d),
\end{align*}
with $\bfy_N=\bfx_N$, and the two innovation sequences $\{\bfep_k\}_{k=0}^{N-1}$ and $\{\bfzeta_k\}_{k=0}^{N-1}$ are independent.

The log-ratio decomposes as
\begin{align*}
    \Delta_h
    =
    \sum_{k=0}^{N-1}\Delta_k^{\mathrm f}
    +
    \sum_{k=0}^{N-1}\Delta_k^{\mathrm b},
\end{align*}
where
\begin{align*}
    \Delta_k^{\mathrm f}
    &:={}
    \log\frac{\leftvec{p}_k(\bfx_k\mid\bfx_{k+1})}
                  {\vec{p}_k(\bfx_{k+1}\mid\bfx_k)}
    +L_{k+1}(\bfx_{k+1})-L_k(\bfx_k),
    \\
    \Delta_k^{\mathrm b}
    &:={}
    \log\frac{\vec{p}_k(\bfy_{k+1}\mid\bfy_k)}
                  {\leftvec{p}_k(\bfy_k\mid\bfy_{k+1})}
    +L_k(\bfy_k)-L_{k+1}(\bfy_{k+1}).
\end{align*}

For a symmetric matrix $A\in\R^{d\times d}$, define
\[
    Q_A(\bfxi)
    :=\bfxi\tran A\bfxi-\tr(A).
\]
In Lemma~\ref{lem:exact-score-local-expansion},
we perform local Taylor expansions of $\Delta_k^{\mathrm f}$ and $\Delta_k^{\mathrm b}$, and decompose $\Delta_h$ into a martingale term and a remainder term as
\begin{align*}
    \Delta_h = M + R,
\end{align*}
where
\begin{equation}\label{eq:def:M:R}
\begin{aligned}
    M&:=M^{\mathrm f}+M^{\mathrm b},\quad
    M^{\mathrm f}:=h\sum_{k=0}^{N-1}Q_{A_k(\bfx_k)}(\bfep_k),\quad
    M^{\mathrm b}:=-h\sum_{k=0}^{N-1}Q_{A_k(\bfy_{k+1})}(\bfzeta_k),\\
    R&:=\Delta_h-M.
\end{aligned}
\end{equation}
In Lemma~\ref{lem:exact-score-global-expansion}, we take expectations and show that the martingale term dominates the remainder term, giving
\begin{equation}\label{equ:M:R:moments}
\begin{aligned}
    &\mathbb E[M] = 0,\quad \operatorname{Var}(M) = h\mathcal D(T) + O(h^{3/2}),\\
    &\|M\|_{L^4} = O(h^{1/2}),\quad \|R\|_{L^2}+\|R\|_{L^4} = O(h).
\end{aligned}
\end{equation}
The Cauchy--Schwarz inequality gives
$|\operatorname{Cov}(M,R)|=O(h^{3/2})$.
Therefore,
\begin{align*}
    \operatorname{Var}(\Delta_h)
    &=
    \operatorname{Var}(M)
    +2\operatorname{Cov}(M,R)
    +\operatorname{Var}(R)
    =
    \operatorname{Var}(M)+O(h^{3/2}).
\end{align*}
This proves that
\[
    \operatorname{Var}(\Delta_h)= h\mathcal D(T) + O(h^{3/2}).
\]
It remains to prove $\EE{\Delta_h} = -\frac12\var(\Delta_h) + O(h^{3/2})$.

Taylor expansion of $e^{\Delta_h}-1$ to third order gives
\begin{align*}
    e^{\Delta_h} - 1 = \Delta_h + \frac12\Delta_h^2 + \frac16\Delta_h^3 + R_4(\Delta_h),
    \qquad
    |R_4(\Delta_h)|
    \leq
    C(1+e^{\Delta_h})|\Delta_h|^4.
\end{align*}
Since $\EE{e^{\Delta_h}}=1$, we have
\begin{align}\label{equ:taylo:expansion:exp:delta}
    0
    &=
    \mathbb E[e^{\Delta_h}-1] 
    =
    \mathbb E[\Delta_h]
    +\frac12\mathbb E[\Delta_h^2]
    +\frac16\mathbb E[\Delta_h^3]
    +O\!\left(
        \mathbb E[
            (1+e^{\Delta_h})|\Delta_h|^4
        ]
      \right).
\end{align}
From Equation~\eqref{equ:M:R:moments}, we have
\[
    \mathbb E|\Delta_h|^4=O(h^2),
\]
and hence by Cauchy--Schwarz,
\[
    \EE{\Delta_h^2} =O(h).
\]
For any integrable function $\psi$, we have $\EE{e^{\Delta_h}\psi(\Delta_h)}=\EE{\psi(-\Delta_h)}$ by Lemma~\ref{lem:exact-score-path-reversal}.
Taking $\psi(x)=|x|^4$ gives
\[
    \EE{e^{\Delta_h}|\Delta_h|^4}
    =
    \EE{|\Delta_h|^4}
    =
    O(h^2).
\]
Taking $\psi(x)=|x|^2$ gives
\begin{align*}
    0
    &=
    \EE{\Delta_h^2(e^{\Delta_h}-1)} 
    =
    \EE{\Delta_h^3}
    +O\left(
        \EE{(1+e^{\Delta_h})|\Delta_h|^4}
      \right),
\end{align*}
so that
\[
    \EE{\Delta_h^3}=O(h^2).
\]
Plugging $\EE{\Delta_h^3}=O(h^2), \EE{(1+e^{\Delta_h} )|\Delta_h|^4}=O(h^2)$ into Equation~\eqref{equ:taylo:expansion:exp:delta} gives
\[
    \EE{\Delta_h}
    =
    -\frac12\EE{\Delta_h^2}+O(h^2).
\]
This together with $\EE{\Delta_h^2}=O(h)$ implies that $|\EE{\Delta_h}| = O(h)$. Therefore,
\[
    \EE{\Delta_h^2} = \var(\Delta_h) + |\EE{\Delta_h}|^2
    =\var(\Delta_h)+O(h^2).
\]
Combining the last two displays gives
\[
    \EE{\Delta_h}
    =
    -\frac12\operatorname{Var}(\Delta_h)+O(h^2) = -\frac{h}{2}\calD(T)+O(h^{3/2}).
\]

\end{proof}

\begin{lem}[Local expansion]
\label{lem:exact-score-local-expansion}
Let
\[
    A_k(\bfx)
    :=
    \nabla f_k(\bfx)
    -\frac{g_k^2}{2}\nabla s_k(\bfx),
    \qquad
    Q_A(\bfxi)
    :=
    \bfxi\tran A\bfxi-\tr(A).
\]
Then
\begin{align*}
    \Delta_k^{\mathrm f}
    &=
    hQ_{A_k(\bfx_k)}(\bfep_k)
    +h^{3/2}U_k^{\mathrm f} + O_{L^4}(h^2),\\
    \Delta_k^{\mathrm b}
    &=
    -hQ_{A_k(\bfy_{k+1})}(\bfzeta_k)
    +h^{3/2}U_k^{\mathrm b} + O_{L^4}(h^2),
\end{align*}
and
\begin{equation}
\label{eq:exact-score-U-centered}
    \mathbb E[U_k^{\mathrm f}\mid\bfx_k]=0,
    \qquad
    \mathbb E[U_k^{\mathrm b}\mid\bfy_{k+1}]=0.
\end{equation}
\end{lem}

\begin{proof}[Proof of Lemma~\ref{lem:exact-score-local-expansion}]
For a scalar function $H$ and a vector field $v$, write
$D^rH(\bfx)[\mathbf u_1,\ldots,\mathbf u_r]$ and
$D^rv(\bfx)[\mathbf u_1,\ldots,\mathbf u_r]$ for their $r$th spatial derivatives, with the latter interpreted componentwise. All remainder
estimates below are uniform in $k$ and hold in every fixed finite
$L^q$ under the stated regularity and moment assumptions in Assumption~\ref{assump:regularity}.

\paragraph{Forward path.}
Fix $k$ and write
\[
    \bfx:=\bfx_k,
    \qquad
    \bfx^+:=\bfx_{k+1},
    \qquad
    \bfep:=\bfep_k.
\]
Suppress the index $k$ on $f_k$, $g_k$, $b_k^\star$, $s_k$, and $L_k$.
Thus,
\[
    \bfx^+
    =
    \bfx+h f(\bfx)+\sqrt h\,g\bfep.
\]
Define
\[
    \widetilde{\bfzeta}
    :=
    \frac{\bfx-\bfx^+-h b^\star(\bfx^+)}
         {g\sqrt h}.
\]
Then
\[
    \Delta_k^{\mathrm f}
    =
    \frac12
    \left(
        \|\bfep\|^2-\|\widetilde{\bfzeta}\|^2
    \right)
    +L_{k+1}(\bfx^+)-L_k(\bfx).
\]
Write
\[
    \widetilde{\bfzeta}=-(\bfep+\bfr),
    \qquad
    \bfr
    :=
    \frac{\sqrt h}{g}
    \left\{
        f(\bfx)+b^\star(\bfx^+)
    \right\}.
\]
Then
\begin{equation}
\label{eq:exact-score-Gaussian-ratio}
    \frac12
    \left(
        \|\bfep\|^2-\|\widetilde{\bfzeta}\|^2
    \right)
    =
    -\bfep\tran\bfr-\frac12\|\bfr\|^2.
\end{equation}

Using $b^\star=-f+g^2s$ and expanding $b^\star(\bfx^+)$ around
$\bfx$ gives
\begin{align*}
    \bfr
    ={}&
    \sqrt h\,g s(\bfx)
    +h\nabla b^\star(\bfx)\bfep
    \notag\\
    &+h^{3/2}\left\{
        \frac1g\nabla b^\star(\bfx)f(\bfx)
        +\frac g2
          D^2b^\star(\bfx)[\bfep,\bfep]
    \right\}
    +O_{L^8}(h^2).
\end{align*}
Substitution into~\eqref{eq:exact-score-Gaussian-ratio} yields
\begin{equation}
\label{eq:exact-score-Gaussian-forward-expansion}
\begin{aligned}
\frac12
    \left(
        \|\bfep\|^2-\|\widetilde{\bfzeta}\|^2
    \right)
&=-\sqrt h\,g s(\bfx)\tran\bfep\\
&\quad+h\left\{
    -\bfep\tran\nabla b^\star(\bfx)\bfep
    -\frac{g^2}{2}\|s(\bfx)\|^2
\right\}
\\
&\quad+h^{3/2}\left\{
    -\frac1g\bfep\tran
        \nabla b^\star(\bfx)f(\bfx)
    -\frac g2\bfep\tran
        D^2b^\star(\bfx)[\bfep,\bfep]
    -g s(\bfx)\tran
        \nabla b^\star(\bfx)\bfep
\right\}\\
&\quad+O_{L^4}(h^2).
\end{aligned}
\end{equation}
Recall that $L_{k}=L_{t_k}$ and $s_t=\nabla L_t$.
A space-time Taylor expansion gives
\begin{equation}
\label{eq:exact-score-ell-expansion}
\begin{aligned}
    L_{k+1}(\bfx^+)-L_k(\bfx)
    ={}&
    \sqrt h\,g s(\bfx)\tran\bfep\\
    &+h\left\{
        \partial_tL(t_k,\bfx)
        +f(\bfx)\tran s(\bfx)
        +\frac{g^2}{2}
          D^2L_k(\bfx)[\bfep,\bfep]
    \right\}\\
    &+h^{3/2}\left\{
        gD^2L_k(\bfx)[\bfep,f(\bfx)]
        +g(\partial_t s)(t_k,\bfx)\tran\bfep
        +\frac{g^3}{6}
          D^3L_k(\bfx)[\bfep,\bfep,\bfep]
    \right\}\\
    &+O_{L^4}(h^2).
\end{aligned}
\end{equation}
Since $p_t$ satisfies the Fokker--Planck equation 
\begin{align*}
\partial_t p_t = -\nabla \cdot (f_t p_t) + \frac{g_t^2}{2} \Delta p_t = -\nabla \cdot (f_t p_t) + \frac{g_t^2}{2} \nabla \cdot (\nabla p_t),
\end{align*} 
we have
\begin{equation}
\label{eq:exact-score-log-Fokker-Planck}
    \partial_tL_t
    =\frac{\partial_t p_t}{p_t}=
    -\nabla\!\cdot f_t
    -f_t\tran s_t
    +\frac{g_t^2}{2}
      \left\{
          \nabla\!\cdot s_t+\|s_t\|^2
      \right\}.
\end{equation}
The order-$h^{1/2}$ terms in
\eqref{eq:exact-score-Gaussian-forward-expansion} and
\eqref{eq:exact-score-ell-expansion} cancel. Using
$\nabla b^\star = -\nabla f+g^2\nabla s$, $D^2L_k=\nabla s$,
and~\eqref{eq:exact-score-log-Fokker-Planck}, the coefficient of $h$ in $\Delta_k^{\mathrm f}$ is
\begin{align*}
    &-\bfep\tran\nabla b^\star(\bfx)\bfep
      -\frac{g^2}{2}\|s(\bfx)\|^2
      +\partial_tL(t_k,\bfx)
      +f(\bfx)\tran s(\bfx)
      +\frac{g^2}{2}
        \bfep\tran\nabla s(\bfx)\bfep\\
    &\qquad
    =
    \bfep\tran
    \left(
        \nabla f(\bfx)
        -\frac{g^2}{2}\nabla s(\bfx)
    \right)\bfep
    -
    \tr\left(
        \nabla f(\bfx)
        -\frac{g^2}{2}\nabla s(\bfx)
    \right)\\
    &\qquad
    =
    Q_{A_k(\bfx_k)}(\bfep_k).
\end{align*}
The coefficient of $h^{3/2}$ is
\begin{align*}
    U_k^{\mathrm f}
    ={}&
    -\frac1{g}\bfep\tran
        \nabla b^\star(\bfx)f(\bfx)
    -\frac{g}{2}\bfep\tran
        D^2b^\star(\bfx)[\bfep,\bfep]
    \notag\\
    &-g s(\bfx)\tran
        \nabla b^\star(\bfx)\bfep
    +gD^2L_k(\bfx)
        [\bfep,f(\bfx)]
    \notag\\
    &+g(\partial_t s)(t_k,\bfx)\tran\bfep
    +\frac{g^3}{6}
        D^3L_k(\bfx)
        [\bfep,\bfep,\bfep].
\end{align*}
Conditionally on $\bfx_k$, each term in
the last display is either linear in $\bfep_k$ or homogeneous cubic in $\bfep_k$.
Since $\bfep_k$ is conditionally standard Gaussian given $\bfx_k$, we have
\begin{align*}
    \mathbb E[U_k^{\mathrm f}\mid\bfx_k]=0.
\end{align*}
Collecting the preceding expansions gives
\[
    \Delta_k^{\mathrm f}
    =
    hQ_{A_k(\bfx_k)}(\bfep_k)
    +h^{3/2}U_k^{\mathrm f}
    +O_{L^4}(h^2).
\]

\paragraph{Backward path.}
Write
\[
    \bfy^+:=\bfy_{k+1},
    \qquad
    \bfy:=\bfy_k,
    \qquad
    \bfzeta:=\bfzeta_k,
\]
and again suppress the index $k$ on the coefficients. Then
\[
    \bfy
    =
    \bfy^+
    +h b^\star(\bfy^+)
    +\sqrt h\,g\bfzeta.
\]
Define
\[
    \widetilde{\bfep}
    :=
    \frac{\bfy^+-\bfy-hf(\bfy)}{g\sqrt h}.
\]
Then
\[
    \widetilde{\bfep}
    =
    -(\bfzeta+\widetilde{\bfr}),
    \qquad
    \widetilde{\bfr}
    :=
    \frac{\sqrt h}{g}
    \left\{
        b^\star(\bfy^+)+f(\bfy)
    \right\},
\]
and
\begin{align*}
    \Delta_k^{\mathrm b}
    &=\frac12\left( \|\bfzeta\|^2  - \|\widetilde{\bfep}\|^2 \right) +L_k(\bfy)-L_{k+1}(\bfy^+)\\
    &=-\bfzeta\tran\widetilde{\bfr}
    -\frac12\|\widetilde{\bfr}\|^2
    +L_k(\bfy)-L_{k+1}(\bfy^+).
\end{align*}
Using $b^\star=-f+g^2s$ and expanding $f(\bfy)$ around $\bfy^+$ gives
\begin{align*}
    \widetilde{\bfr}
    ={}&
    \sqrt h\,g s(\bfy^+)
    +h\nabla f(\bfy^+)\bfzeta
    \notag\\
    &+h^{3/2}\left\{
        \frac1g\nabla f(\bfy^+)b^\star(\bfy^+)
        +\frac g2
          D^2f(\bfy^+)[\bfzeta,\bfzeta]
    \right\}
    +O_{L^8}(h^2).
\end{align*}
Consequently,
\begin{equation}\label{eq:exact-score-Gaussian-backward-expansion}
\begin{aligned}
    -\bfzeta\tran\widetilde{\bfr}
    -\frac12\|\widetilde{\bfr}\|^2
    ={}&
    -\sqrt h\,g s(\bfy^+)\tran\bfzeta\\
    &+h\left\{
        -\bfzeta\tran\nabla f(\bfy^+)\bfzeta
        -\frac{g^2}{2}\|s(\bfy^+)\|^2
    \right\}\\
    &+h^{3/2}\left\{
        -\frac1g\bfzeta\tran
            \nabla f(\bfy^+)b^\star(\bfy^+)
        -\frac g2\bfzeta\tran
            D^2f(\bfy^+)[\bfzeta,\bfzeta]
        -g s(\bfy^+)\tran
            \nabla f(\bfy^+)\bfzeta
    \right\}\\
    &+O_{L^4}(h^2).
\end{aligned}
\end{equation}
A space-time Taylor expansion gives
\begin{equation}\label{eq:exact-score-ell-backward-expansion}
\begin{aligned}
    L_k(\bfy)-L_{k+1}(\bfy^+)
    ={}&
    \sqrt h\,g s(\bfy^+)\tran\bfzeta\\
    &+h\left\{
        -\partial_tL(t_k,\bfy^+)
        +b^\star(\bfy^+)\tran s(\bfy^+)
        +\frac{g^2}{2}
          D^2L_k(\bfy^+)[\bfzeta,\bfzeta]
    \right\}\\
    &+h^{3/2}\left\{
        gD^2L_k(\bfy^+)
            [\bfzeta,b^\star(\bfy^+)]
        +\frac{g^3}{6}
          D^3L_k(\bfy^+)
            [\bfzeta,\bfzeta,\bfzeta]
    \right\}\\
    &+O_{L^4}(h^2).
\end{aligned}
\end{equation}
The order-$h^{1/2}$ terms in~\eqref{eq:exact-score-Gaussian-backward-expansion} and~\eqref{eq:exact-score-ell-backward-expansion} cancel. 
Using
$b^\star=-f+g^2s$ and the log-Fokker--Planck equation~\eqref{eq:exact-score-log-Fokker-Planck}, 
the order-$h$ coefficient in $\Delta_k^{\mathrm b}$ is
\begin{align*}
    &-\bfzeta\tran\nabla f(\bfy^+)\bfzeta
      -\frac{g^2}{2}\|s(\bfy^+)\|^2
      -\partial_tL(t_k,\bfy^+)
      +b^\star(\bfy^+)\tran s(\bfy^+)
      +\frac{g^2}{2}
        \bfzeta\tran\nabla s(\bfy^+)\bfzeta\\
    &\qquad
    =
    -\bfzeta\tran
    \left(
        \nabla f(\bfy^+)
        -\frac{g^2}{2}\nabla s(\bfy^+)
    \right)\bfzeta
    +
    \tr\left(
        \nabla f(\bfy^+)
        -\frac{g^2}{2}\nabla s(\bfy^+)
    \right)\\
    &\qquad
    =
    -Q_{A_k(\bfy_{k+1})}(\bfzeta_k).
\end{align*}
The coefficient of $h^{3/2}$ is
\begin{align*}
    U_k^{\mathrm b}
    ={}&
    -\frac1{g}\bfzeta\tran
        \nabla f(\bfy^+)b^\star(\bfy^+)
    -\frac{g}{2}\bfzeta\tran
        D^2f(\bfy^+)[\bfzeta,\bfzeta]
    \notag\\
    &-g s(\bfy^+)\tran
        \nabla f(\bfy^+)\bfzeta
    +g D^2L_k(\bfy^+)
        [\bfzeta,b^\star(\bfy^+)]
    \notag\\
    &+\frac{g^3}{6}
        D^3L_k(\bfy^+)
        [\bfzeta,\bfzeta,\bfzeta].
\end{align*}
Conditionally on $\bfy^+$, each term in
the last display is either linear in $\bfzeta_k$ or homogeneous cubic in $\bfzeta_k$.
Since $\bfzeta_k$ is conditionally standard Gaussian given $\bfy_{k+1}$, therefore,
\begin{align*}
    \mathbb E[U_k^{\mathrm b}\mid\bfy_{k+1}]&=0.
\end{align*}
Thus,
\[
    \Delta_k^{\mathrm b}
    =
    -hQ_{A_k(\bfy_{k+1})}(\bfzeta_k)
    +h^{3/2}U_k^{\mathrm b}
    +O_{L^4}(h^2).
\]

\end{proof}

\begin{lem}[Leading martingale and accumulated remainder]
\label{lem:exact-score-global-expansion}
Let $M$ and $R$ be defined as in Equation~\eqref{eq:def:M:R}. Then
\begin{align}
    \mathbb E[M]&=0,
    \label{eq:exact-score-M-mean}\\
    \operatorname{Var}(M)
    &=h\mathcal D(T)+O(h^{3/2}),
    \label{eq:exact-score-M-variance}\\
    \|M\|_{L^4}&=O(h^{1/2}),
    \qquad
    \|R\|_{L^2}+\|R\|_{L^4}=O(h).
    \label{eq:exact-score-MR-moments}
\end{align}
Consequently, $\mathbb E|\Delta_h|^4=O(h^2)$.
\end{lem}
\begin{proof}[Proof of Lemma~\ref{lem:exact-score-global-expansion}]

Define $\mathcal F_k:=\sigma(\bfx_0,\ldots,\bfx_k)$ for
$k=0,\ldots,N$. After the forward path has been generated, extend this
filtration by revealing the backward states in the order
$\bfy_{N-1},\bfy_{N-2},\ldots,\bfy_0$. Since $\bfep_k$ is independent of
$\mathcal F_k$,
\[
    \mathbb E\!\left[
        Q_{A_k(\bfx_k)}(\bfep_k)\mid\mathcal F_k
    \right]=0
\]
by Lemma~\ref{lem:exact-score-elementary-moments}. Thus, the summands of
$M^{\mathrm f}$ are martingale differences. Similarly, when the backward
summands are ordered from $k=N-1$ down to $k=0$, $\bfzeta_k$ is independent
of the current sigma-field and of $\bfy_{k+1}$. Hence, those summands are
also martingale differences, and
$\mathbb E[M^{\mathrm f}]=\mathbb E[M^{\mathrm b}]=0$.
This proves \eqref{eq:exact-score-M-mean}.
Using the martingale property and Lemma~\ref{lem:exact-score-elementary-moments},
\begin{align}
    \operatorname{Var}(M^{\mathrm f})
    &=
    2h^2\sum_{k=0}^{N-1}
    \mathbb E\!\left[\|A_k(\bfx_k)\|_{\mathrm F}^2\right],
    \label{eq:exact-score-Mf-variance}\\
    \operatorname{Var}(M^{\mathrm b})
    &=
    2h^2\sum_{k=0}^{N-1}
    \mathbb E\!\left[\|A_k(\bfy_{k+1})\|_{\mathrm F}^2\right].
    \label{eq:exact-score-Mb-variance}
\end{align}
Conditionally on the complete forward path, the backward path starts from the fixed endpoint $\bfy_N=\bfx_N$ and is driven by independent innovations. Its leading martingale therefore has conditional mean zero, so
\begin{equation}
\label{eq:exact-score-Mf-Mb-covariance}
    \operatorname{Cov}(M^{\mathrm f},M^{\mathrm b})=0.
\end{equation}

Let \(G_t(\bfx):=\|A_t(\bfx)\|_{\mathrm F}^2\). By the polynomial-growth
assumptions on \(A_t\), \(G_t\) satisfies,
for some \(q\ge 0\),
\[
    |G_t(\bfx)-G_t(\bfz)|
    \le
    C(1+\|\bfx\|^q+\|\bfz\|^q)\|\bfx-\bfz\|.
\]
Together with the uniform moment bounds and the strong error estimate
\[
    \max_{0\le k\le N}\|\bfx_k-X_{t_k}\|_{L^q}\le C h^{1/2},
\]
this implies
\[
    \big|
    \E G_{t_k}(\bfx_k)-\E G_{t_k}(X_{t_k})
    \big|
    \le C h^{1/2}
\]
uniformly in \(k\). 
Moreover, \(t\mapsto \E G_t(X_t)\) is \(1/2\)-H\"older continuous under the
same assumptions, since \(\|X_t-X_s\|_{L^p}\le C|t-s|^{1/2}\). Therefore,
\[
    h\sum_{k=0}^{N-1}\E G_{t_k}(X_{t_k})
    =
    \int_0^T \E G_t(X_t)\,\rd t
    +O(h^{1/2}).
\]
Combining the last two displays gives
\begin{align}\label{eq:exact-score-forward-Riemann}
    h\sum_{k=0}^{N-1}\mathbb E[G_{t_k}(\bfx_k)]
    =
    \int_0^T\mathbb E[G_t(X_t)]\,\rd t
    +O(h^{1/2}).
\end{align}
Following the same argument and that the exact reverse process has marginal density $p_t$, 
we obtain
\begin{align}\label{eq:exact-score-backward-Riemann}
    h\sum_{k=0}^{N-1}\mathbb E[G_{t_k}(\bfy_{k+1})]
    =
    \int_0^T\mathbb E[G_t(X_t)]\,\rd t
    +O(h^{1/2}).
\end{align}
Equations
\eqref{eq:exact-score-Mf-variance}--\eqref{eq:exact-score-backward-Riemann} prove
\eqref{eq:exact-score-M-variance}.

Since $M^{\mathrm f}$ and $M^{\mathrm b}$ are sum of martingale differences with $O_{L^4}(h)$ increments,
Lemma~\ref{lem:exact-score-elementary-moments} gives
\[
    \|M^{\mathrm f}\|_{L^4}
    +\|M^{\mathrm b}\|_{L^4}
    =O(h^{1/2}).
\]
By Lemma~\ref{lem:exact-score-local-expansion},
$R$ is the sum of the martingale difference terms $h^{3/2}U_k^{\mathrm f}$, $h^{3/2}U_k^{\mathrm b}$, as well as the $O_{L^4}(h^2)$ remainder term.
Applying Lemma~\ref{lem:exact-score-elementary-moments} again gives
\[
    \Big\|
        h^{3/2}\sum_{k=0}^{N-1}U_k^{\mathrm f}
    \Big\|_{L^2}
    +
    \Big\|
        h^{3/2}\sum_{k=0}^{N-1}U_k^{\mathrm b}
    \Big\|_{L^4}
    =O(h).
\]
Applying Minkowski's inequality to the
$N$ local $O_{L^4}(h^2)$ remainder terms gives
$\|R\|_{L^4}=O(h)$.  Finally,
$\|\Delta_h\|_{L^4}\leq\|M\|_{L^4}+\|R\|_{L^4}=O(h^{1/2})$, which proves
$\mathbb E|\Delta_h|^4=O(h^2)$.
\end{proof}

\begin{lem}[Path-reversal identity]
\label{lem:exact-score-path-reversal}
For an integrable nonnegative function $\psi$,
\begin{align*}
    \mathbb E\!\left[e^{\Delta_h}\psi(\Delta_h)\right]
    =\mathbb E[\psi(-\Delta_h)].
\end{align*}
\end{lem}
\begin{proof}[Proof of Lemma~\ref{lem:exact-score-path-reversal}]
Recall the notation from Section~\ref{sec:algorithm}.
We have
\[
    \Delta_h(\bfz)
    =\log\frac{\Pi(\mathsf R\bfz)}{\Pi(\bfz)},
    \qquad
    \Delta_h(\mathsf R\bfz)=-\Delta_h(\bfz).
\]
Therefore,
\begin{align*}
    \mathbb E\!\left[e^{\Delta_h}\psi(\Delta_h)\right]
    &=
    \int \Pi(\mathsf R\bfz)\psi(\Delta_h(\bfz))\,\rd \bfz 
    =\int \Pi(\bfz)\psi(-\Delta_h(\bfz))\,\rd \bfz,
\end{align*}
where the second equality uses the change of variables
$\bfz\mapsto\mathsf R\bfz$. This proves the identity.
\end{proof}

\begin{lem}[Elementary Gaussian and martingale moments]
\label{lem:exact-score-elementary-moments}
Let $\bfxi\sim\mathcal N(\mathbf{0},I_d)$ and let $A$ be symmetric.
Then
\begin{equation}
\label{eq:exact-score-Q-moments}
    \mathbb E[Q_A(\bfxi)]=0,
    \qquad
    \mathbb E[Q_A(\bfxi)^2]
    =2\|A\|_{\mathrm F}^2,
\end{equation}
and, for fixed $d$,
\begin{equation}
\label{eq:exact-score-Q-fourth}
    \mathbb E|Q_A(\bfxi)|^4
    \leq C_d\|A\|_{\mathrm F}^4.
\end{equation}

In addition, let $(\mathcal F_k)_{k=0}^N$ be a filtration and let
$\eta_k$ be $\mathcal F_{k+1}$-measurable with
$\mathbb E[\eta_k\mid\mathcal F_k]=0$ and
$\max_k\|\eta_k\|_{L^4}\leq a$. Then
\begin{equation}
\label{eq:exact-score-martingale-moments}
    \left\|\sum_{k=0}^{N-1}\eta_k\right\|_{L^2}
    \leq N^{1/2}a,
    \qquad
    \left\|\sum_{k=0}^{N-1}\eta_k\right\|_{L^4}
    \leq C N^{1/2}a,
\end{equation}
where $C$ is a universal constant.
\end{lem}
\begin{proof}[Proof of Lemma~\ref{lem:exact-score-elementary-moments}]
Diagonalizing $A$ and using rotational invariance gives
\[
    Q_A(\bfxi)
    \stackrel{d}{=}
    \sum_{i=1}^d\lambda_i(Z_i^2-1),
    \qquad
    Z_i\stackrel{\mathrm{i.i.d.}}{\sim}\mathcal N(0,1).
\]
This proves \eqref{eq:exact-score-Q-moments}; expanding the fourth power
proves \eqref{eq:exact-score-Q-fourth}.

The $L^2$ martingale bound follows from orthogonality.  For the fourth
moment, set $S_k=\sum_{j=0}^{k-1}\eta_j$.  Expanding
$(S_k+\eta_k)^4$ and conditioning on $\mathcal F_k$ gives
\begin{align*}
    \mathbb E[S_{k+1}^4]
    ={}&
    \mathbb E[S_k^4]
    +6\mathbb E\!\left[
        S_k^2\mathbb E(\eta_k^2\mid\mathcal F_k)
    \right]
    +4\mathbb E\!\left[
        S_k\mathbb E(\eta_k^3\mid\mathcal F_k)
    \right]
    +\mathbb E[\eta_k^4].
\end{align*}
Jensen's inequality gives
\[
    \left\|\mathbb E(\eta_k^2\mid\mathcal F_k)\right\|_{L^2}
    \leq a^2,
    \qquad
    \left\|\mathbb E(\eta_k^3\mid\mathcal F_k)\right\|_{L^{4/3}}
    \leq a^3.
\]
H\"older's inequality therefore gives
\[
    \mathbb E[S_{k+1}^4]
    \leq
    \mathbb E[S_k^4]
    +6a^2\|S_k^2\|_{L^2}
    +4a^3\|S_k\|_{L^4}
    +a^4.
\]
A direct induction, with the constant chosen sufficiently large, yields
$\mathbb E[S_k^4]\leq Ck^2a^4$.  This proves the second bound in
\eqref{eq:exact-score-martingale-moments}.
\end{proof}

\begin{rmk}[Exponential integrator]
\label{rmk:exponential_integrator}
For the OU process, consider the exact forward update and the
exponential backward update
\begin{align*}
    \bfx_{k+1}
    &=e^{-h/2}\bfx_k+\sqrt{1-e^{-h}}\,\bfep_k,\\
    \bfy_k
    &=e^{h/2}\bfy_{k+1}
      +2(e^{h/2}-1)s_{t_k}(\bfy_{k+1})
      +\sqrt{e^h-1}\,\bfzeta_k,
\end{align*}
with the corresponding Gaussian transition densities used in the
Metropolis--Hastings ratio.
The same local-expansion argument applies, with
\[
    A_t(\bfx)=-\frac12\nabla s_t(\bfx).
\]
Indeed, the forward and backward variances satisfy
\[
    1-e^{-h}=h-\frac12h^2+O(h^3),
    \qquad
    e^h-1=h+\frac12h^2+O(h^3).
\]
Their difference contributes
$\frac h2(\|\bfep_k\|^2-d)$ to the leading forward
log-ratio, with the opposite sign in the backward expansion.
This replaces the Euler coefficient
$-\frac12(I_d+\nabla s_t)$ by $-\frac12\nabla s_t$.
Consequently, Theorem~\ref{thm:exact-score-discretization} holds with
$\mathcal D(T)$ replaced by
\[
    \mathcal D_{\mathrm{EI}}(T)
    :=\int_0^T
    \mathbb E\!\left[
        \|\nabla s_t(X_t)\|_{\mathrm F}^2
    \right]\,\rd t.
\]
\end{rmk}

\subsection{Proof of Theorem~\ref{thm:score-error-mean}}
\label{prf:thm:score-error-mean}
\begin{proof}[Proof of Theorem~\ref{thm:score-error-mean}]
We use the same notation as in the proof of Theorem~\ref{thm:exact-score-discretization}. Specifically, write
\[
    f_k:=f_{t_k},\quad
    g_k:=g_{t_k},\quad
    s_k:=s_{t_k},\quad
    \delta_k:=\delta_{t_k},\quad
    \widehat s_k:=s_k+\delta_k,
    \quad
    \widehat b_k:=-f_k+g_k^2\widehat s_k,
\]
and
\[
    A_k(\bfx)
    :=\nabla f_k(\bfx)-\frac{g_k^2}{2}\nabla s_k(\bfx),
    \quad
    Q_A(\bfxi):=\bfxi\tran A\bfxi-\tr(A).
\]

\paragraph{Forward path.}
Fix \(k\), condition on \(\bfx_k\), and write
\[
    \bfx^+:=\bfx_{k+1},\qquad
    \bfx:=\bfx_k,\qquad
    \bfep:=\bfep_k,\qquad
    \widetilde{\bfzeta}
    :=
    \frac{\bfx-\bfx^+-h \widehat{b}(\bfx^+)}
         {g\sqrt h}.
\] 
Suppress the index $k$ in the subsequent calculations.
We have
\[
    \Delta_k^{\mathrm f}
    =
    \frac12
    \left(
        \|\bfep\|^2-\|\widetilde{\bfzeta}\|^2
    \right)
    +L_{k+1}(\bfx^+)-L_k(\bfx).
\]
As in the proof of Lemma~\ref{lem:exact-score-local-expansion}, we take the Taylor expansions around $\bfx$, with $b^\star$ replaced by $\widehat b$, and obtain
\begin{align*}
 \frac12\left(
        \|\bfep\|^2 - \|\widetilde{\bfzeta}\|^2
    \right)
 ={}&-\sqrt h\,g\widehat s\tran\bfep \\
 &+h\left\{
      -\bfep\tran \nabla\widehat b\bfep
      -\frac{g^2}{2}\|\widehat s\|^2
    \right\} \\
 &+h^{3/2}\left\{
      -\frac1{g}\bfep\tran \nabla\widehat b f
      -\frac{g}{2}\bfep\tran
          D^2\widehat b[\bfep,\bfep]
      -g\widehat s\tran(\nabla\widehat b)\bfep
    \right\}\\
    &+O_{L^8}(h^2)
\end{align*}
as well as
\begin{align*}
 L_{k+1}(\bfx_{k+1})-L_k(\bfx_k)
 ={}&\sqrt h\,g s\tran\bfep \\
 &+h\left\{
      \partial_tL_k+f\tran s
      +\frac{g^2}{2}\bfep\tran(\nabla s)\bfep
    \right\} \\
 &+h^{3/2}\left\{
      g\bfep\tran(\partial_t s)(t_k,\bfx)
      +g\bfep\tran(\nabla s)f
      +\frac{g^3}{6}
         \bfep\tran D^2s[\bfep,\bfep]
    \right\}\\
    &+O_{L^8}(h^2).
\end{align*}
Since $s_k-\widehat s_k = -\delta_k$, the order-\(h^{1/2}\) coefficient in the sum of these two expansions is
\(-g_k\delta_k\tran\bfep\). 
Using
\[
    \nabla\widehat b_k
    =-\nabla f_k+g_k^2\nabla s_k+g_k^2\nabla\delta_k
\]
and the log Fokker--Planck identity~\eqref{eq:exact-score-log-Fokker-Planck},
the order-\(h\) coefficient is
\begin{align*}
&-\bfep\tran \nabla\widehat b_k \bfep
     -\frac{g_k^2}{2}\|s_k+\delta_k\|^2
     +\partial_tL_k+f_k\tran s_k
     +\frac{g_k^2}{2}\bfep\tran(\nabla s_k)\bfep\\
=&\bfep\tran
       \left(
         \nabla f_k-\frac{g_k^2}{2}\nabla s_k
       \right)\bfep
      -\nabla\!\cdot f_k
      +\frac{g_k^2}{2}\nabla\!\cdot s_k 
      -g_k^2\bfep\tran(\nabla\delta_k)\bfep
      -g_k^2s_k\tran\delta_k
      -\frac{g_k^2}{2}\|\delta_k\|^2 \\
=& Q_{A_k(\bfx)}(\bfep)
      -g_k^2\left\{
          s_k\tran\delta_k
          +\bfep\tran(\nabla\delta_k)\bfep
          +\frac12\|\delta_k\|^2
       \right\}.
\end{align*}
In the order-\(h^{3/2}\) coefficient, each term is an odd polynomial in $\bfep_k$,
and hence
\begin{equation}
    \mathbb E[U_k^{\mathrm f}\mid\bfx_k]=0.
    \label{eq:score-error-forward-U-centered}
\end{equation}
Collecting the coefficients gives
\begin{align}
    \Delta_k^{\mathrm f}
    ={}&
    -\sqrt h\,g_k\delta_k(\bfx_k)\tran\bfep_k
    -\frac h2g_k^2\|\delta_k(\bfx_k)\|^2
    +hQ_{A_k(\bfx_k)}(\bfep_k)
    \notag\\
    &-hg_k^2\left\{
        \delta_k(\bfx_k)\tran s_k(\bfx_k)
        +\bfep_k\tran\nabla\delta_k(\bfx_k)\bfep_k
      \right\}
    +h^{3/2}U_k^{\mathrm f}+O_{L^4}(h^2).
    \label{eq:score-error-forward-local}
\end{align}

\paragraph{Backward path.}
Fix \(k\), condition on \(\bfy_{k+1}\), and write
\[
    \bfy^+:=\bfy_{k+1},\quad
    \bfy:=\bfy_k,\quad
    \bfzeta:=\bfzeta_k,\quad
    \widetilde{\bfep}
    :=
    \frac{\bfy^+-\bfy-hf(\bfy)}{g\sqrt h}.
\]
Following a similar argument as in the forward path, we have
\begin{align*}
\frac12\left(
      \|\bfzeta\|^2
      -\|\widetilde{\bfep}\|^2
   \right)
 ={}&-\sqrt h\,g\widehat s_k\tran\bfzeta \\
 &+h\left\{
      -\bfzeta\tran(\nabla f)\bfzeta
      -\frac{g^2}{2}\|\widehat s_k\|^2
    \right\} \\
 &+h^{3/2}\left\{
      -\frac1{g}\bfzeta\tran(\nabla f)\widehat b_k
      -\frac{g}{2}\bfzeta\tran
          D^2f[\bfzeta,\bfzeta]
      -g\widehat s_k\tran(\nabla f)\bfzeta
    \right\}\\
&+O_{L^4}(h^2)
\end{align*}
as well as 
\begin{align*}
 L_k(\bfy_k)-L_{k+1}(\bfy_{k+1})
 ={}&\sqrt h\,g s_k\tran\bfzeta \\
 &+h\left\{
      -\partial_tL_k+\widehat b_k\tran s_k
      +\frac{g^2}{2}\bfzeta\tran(\nabla s_k)\bfzeta
    \right\} \\
 &+h^{3/2}\left\{
      g\bfzeta\tran(\nabla s_k)\widehat b_k
      +\frac{g^3}{6}
         \bfzeta\tran D^2s_k[\bfzeta,\bfzeta]
    \right\}\\
    &+O_{L^4}(h^2).
\end{align*}
Thus, the order-\(h^{1/2}\) coefficient is
\(-g_k\delta_k\tran\bfzeta\), while the order-\(h\) coefficient is
\begin{align*}
&-\bfzeta\tran(\nabla f_k)\bfzeta
     +\frac{g_k^2}{2}\bfzeta\tran(\nabla s_k)\bfzeta
     -\frac{g_k^2}{2}\|s_k+\delta_k\|^2
     -\partial_tL_k+\widehat b_k\tran s_k\\
=&-Q_{A_k(\bfy^+)}(\bfzeta)
  -\frac{g_k^2}{2}\|\delta_k(\bfy^+)\|^2,
\end{align*}
where the last equality follows from \(\widehat b_k=-f_k+g_k^2(s_k+\delta_k)\) and the log Fokker--Planck identity~\eqref{eq:exact-score-log-Fokker-Planck}.

The complete order-\(h^{3/2}\) coefficient is
\begin{align}
 U_k^{\mathrm b}
 ={}&
   -\frac1{g_k}\bfzeta_k\tran
       \nabla f_k(\bfy_{k+1})
       \widehat b_k(\bfy_{k+1})
   -\frac{g_k}{2}\bfzeta_k\tran
       D^2f_k(\bfy_{k+1})[\bfzeta_k,\bfzeta_k]
   \notag\\
 &-g_k\widehat s_k(\bfy_{k+1})\tran
       \nabla f_k(\bfy_{k+1})\bfzeta_k
   +g_k\bfzeta_k\tran
       \nabla s_k(\bfy_{k+1})
       \widehat b_k(\bfy_{k+1})
   \notag\\
 &+\frac{g_k^3}{6}\bfzeta_k\tran
       D^2s_k(\bfy_{k+1})[\bfzeta_k,\bfzeta_k].
 \label{eq:score-error-U-backward-explicit}
\end{align}
As before, the order-\(h^{3/2}\) coefficient is an odd polynomial in \(\bfzeta_k\),
and hence
\begin{equation}
    \mathbb E[U_k^{\mathrm b}\mid\bfy_{k+1}]=0.
    \label{eq:score-error-backward-U-centered}
\end{equation}

Collecting the coefficients gives
\begin{align}
    \Delta_k^{\mathrm b}
    ={}&
    -\sqrt h\,g_k\delta_k(\bfy_{k+1})\tran\bfzeta_k
    -\frac h2g_k^2\|\delta_k(\bfy_{k+1})\|^2
    -hQ_{A_k(\bfy_{k+1})}(\bfzeta_k)
    \notag\\
    &+h^{3/2}U_k^{\mathrm b}+O_{L^4}(h^2).
    \label{eq:score-error-backward-local}
\end{align}

\paragraph{Taking expectations.}
Since \(\bfep_k\) is independent of \(\bfx_k\) and standard Gaussian,
\[
    \mathbb E[
       Q_{A_k(\bfx_k)}(\bfep_k)
       \mid\bfx_k
    ]=0,
    \qquad
    \mathbb E\!\left[
       \bfep_k\tran\nabla\delta_k(\bfx_k)\bfep_k
       \mid\bfx_k
    \right]
    =\nabla\!\cdot\delta_k(\bfx_k).
\]
Summation of~\eqref{eq:score-error-forward-local} gives
\begin{align}
    \mathbb E\!\left[
       \sum_{k=0}^{N-1}\Delta_k^{\mathrm f}
    \right]
    ={}&
    -\frac h2\sum_{k=0}^{N-1}g_k^2
       \mathbb E\!\left[
          \|\delta_k(\bfx_k)\|^2
       \right]
    \notag\\
    &-h\sum_{k=0}^{N-1}g_k^2
       \mathbb E\!\left[
          \delta_k(\bfx_k)\tran s_k(\bfx_k)
          +\nabla\!\cdot\delta_k(\bfx_k)
       \right]
    +O(Nh^2).
    \label{eq:score-error-forward-mean-sum}
\end{align}
By Assumption~\ref{assump:regularity}--\ref{assump:regularity:score}, $\delta_t (\bfx) p_t(\bfx)$ vanishes as $\|\bfx\|\to \infty$, and $\nabla\! \cdot (\delta_tp_t) \in L^1$.
Since $X_t\sim p_t$ and $s_t=\nabla\log p_t$, Stein's identity implies
\begin{align*}
    \mathbb E\!\left[
        \delta_t(X_t)\tran s_t(X_t)
        +\nabla\!\cdot\delta_t(X_t)
    \right]=0.
\end{align*}
The weak-error bound gives
\begin{align*}
    \mathbb E\!\left[
        \delta_k(\bfx_k)\tran s_k(\bfx_k)
        +\nabla\!\cdot\delta_k(\bfx_k)
    \right]=O(h).
\end{align*}
Since \(Nh^2=Th=O(h)\), this 
implies that the second sum in~\eqref{eq:score-error-forward-mean-sum} is \(O(h)\). The same
weak-error bound gives
\begin{equation}
    h\sum_{k=0}^{N-1}g_k^2
       \mathbb E\!\left[
          \|\delta_k(\bfx_k)\|^2
       \right]
    =
    \int_0^T g_t^2
       \mathbb E\!\left[
          \|\delta_t(X_t)\|^2
       \right]\,\rd t
    +O(h).
    \label{eq:score-error-forward-Riemann}
\end{equation}
Consequently,
\begin{equation}
    \mathbb E\!\left[
       \sum_{k=0}^{N-1}\Delta_k^{\mathrm f}
    \right]
    =
    -\frac12\int_0^T g_t^2
       \mathbb E\!\left[
          \|\delta_t(X_t)\|^2
       \right]\,\rd t
    +O(h).
    \label{eq:score-error-forward-mean}
\end{equation}

By the same argument, we have
\begin{equation}
    \mathbb E\!\left[
       \sum_{k=0}^{N-1}\Delta_k^{\mathrm b}
    \right]
    =
    -\frac12\int_0^T g_t^2
       \mathbb E\!\left[
          \|\delta_t(\widehat Y_t)\|^2
       \right]\,\rd t
    +O(h).
    \label{eq:score-error-backward-mean}
\end{equation}
Combining~\eqref{eq:score-error-forward-mean} and~\eqref{eq:score-error-backward-mean} proves
\[
    \mathbb E[\Delta]
    =-\mathcal E_{\mathrm{score}}(\delta,T)+O(h).
\]

To prove the lower bound on the acceptance probability, by Jensen's inequality, we have
\begin{align*}
    \EE{1\wedge e^{\Delta}} = \EE{\exp(\Delta - \max(\Delta, 0) ) } \geq \exp(\EE{\Delta} - \EE{\max(\Delta, 0) } ).
\end{align*}
Since $\EE{\max(\Delta, 0)}\leq \EE{e^{\Delta-1} }=1/e$ and $e^{-1/e} > 1/2 $, this proves that
\begin{align*}
    \EE{1\wedge e^{\Delta} } \geq \frac12 \exp(\EE{\Delta}) = \frac12 \exp(-\mathcal E_{\mathrm{score}}(\delta,T)+O(h)).
\end{align*}
\end{proof}

\section{Experiment supplements} \label{appendix:experiments}

This section provides the experiment details. In particular, we describe the detailed procedure of score learning in Section~\ref{sec:practical} and provide further details of the experiments discussed in Section~\ref{sec:experiments}.

\subsection{Learning the score} \label{appendix:score}

We parametrize the control $u$ in Section~\ref{sec:practical} by a two-layer multi-layer perceptron.

\paragraph{Noising dynamics and reference distribution}
We instantiate the general noising process~\eqref{equ:reference:process} with the OU coefficients $f_t(\bfx)=-\tfrac{1}{2}\bfx$ and $g_t\equiv 1$, so that $X_t\mid X_0 \sim \N(\sqrt{\alpha_t} X_0, \, (1-\alpha_t)I_d)$ with $\alpha_t = e^{-t}$. The reference distribution is taken to be a Gaussian mixture $p^{\mathrm{ref}} = \sum_{k=1}^K \hat w_k\,\mathcal{N}(\hat\mu_k, \hat\Sigma_k)$, fitted as described below.
Because the OU dynamics are linear, the marginal $p^{\mathrm{ref}}_t$ remains a Gaussian mixture with unchanged weights $\hat{w}_k$, means $\sqrt{\alpha_t} \hat{\mu}_k$, and covariances $(1-\alpha_t)I_d + \alpha_t \hat{\Sigma}_k$, so the reference score $s_t^{\mathrm{ref}} = \nabla \log p_t^{\mathrm{ref}}$ is available in closed form. 
Moreover, the OU process approaches $N(0, I_d)$ at a geometric rate irrespective of the initial distribution. 
Therefore, the memoryless condition approximately holds for moderately large $T$.

\paragraph{GMM fitting}
The GMM parameters are estimated in two stages: first, $N$ particles are evolved under an annealed Langevin algorithm with inverse-temperature schedule $0<\beta_1<\cdots < \beta_L=1$ and adaptive step sizes $h_\ell = h/(\beta_\ell+\varepsilon)$, to obtain samples covering all major modes. The resulting samples are then post-processed to extract $K$ prototype centers via gradient ascent and greedy merging: modes located within a given radius are merged into a single mode, and clusters with fewer than a minimal sample count are discarded. Finally, $K$-means is run on the collected samples initialized from the prototype centers, and the GMM parameters $\{\hat w_k, \hat\mu_k, \hat\Sigma_k\}_{k=1}^K$ are read off from the resulting cluster assignments.

\paragraph{Computational cost} 
The cost of target score evaluation typically dominates in score training (as well as sampling). Therefore, we use it as a measure of computational cost. Since the target score is not incorporated into the parametrization of the control $u_\theta$, it is evaluated only during the annealed Langevin warm-up, the gradient ascent step, and once per trajectory during score training. Therefore, the computational cost is given by
\begin{align*}
    \text{computational cost} &= (\text{\# particles})\times (\text{\# temperatures}) \times (\text{\# steps per temperature}) \\
    &\quad + (\text{\# gradient particles}) \times (\text{\# gradient steps}) \\
    &\quad + (\text{\# training particles}) \times (\text{\# training iterations}).
\end{align*}

\subsection{Further details of examples}
\label{app:experiments:details}

We provide additional implementation details and results for the experiments in Section~\ref{sec:experiments}. 

\paragraph{Bayesian Gaussian mixture model (Section~\ref{sec:mixture:model})}

We reparametrize all parameters into unconstrained space. The mixture weights are represented via logits, and each precision matrix $\Sigma_k^{-1}$ is parameterized through the lower-triangular Cholesky factor, with diagonal entries stored on the log scale to enforce positivity. This yields a parameter vector in $\R^{11}$.
Prior hyperparameters are set to $m_0 = \bar{y}$, $V_0 = 9\,\mathrm{diag}(\hat\sigma^2)$, $\nu_0 = p + 2 = 4$, $W_0 = \mathrm{diag}(\hat\sigma^2)$, and $\alpha_0 = 1$, where $\hat\sigma^2$ denotes the per-coordinate sample variance of the data.
Observations are standardized to zero mean and unit variance prior to analysis. 

The wall-clock times (in seconds) for MAD-Path (including score training) and APT-DEO over 100 runs are $27.86 \pm 0.24$ and  $27.90 \pm 0.54$, respectively.

\paragraph{Bayesian mixture of beta-binomial and binomial (Section~\ref{sec:fetal:deaths})} 

Following \cite{tjelmeland2001mode}, we use independent vague priors for all the parameters. In particular, we use uniform distributions on the unit interval for $\gamma$, $\mu$ and $\nu$ and an exponential prior with expectation one for $\theta$, then apply logit transformations on $\gamma, \mu,\nu$ and a log-transformation on $\theta$. The transformed parameters are denoted by $\tilde{\gamma}, \tilde{\mu},\tilde{\nu}$ and $\tilde{\theta}$, respectively.

Figure~\ref{fig:fetal-deaths-traces} shows the trace plots of the transformed parameters under different choices of time horizon $T$ and number of steps $N$. Mode crossings are more likely to happen when $T$ and $N$ are large, which is consistent with the results in Figure~\ref{fig:fetal-deaths}.

\begin{figure}[htbp]
    \centering
    \includegraphics[width=\linewidth]{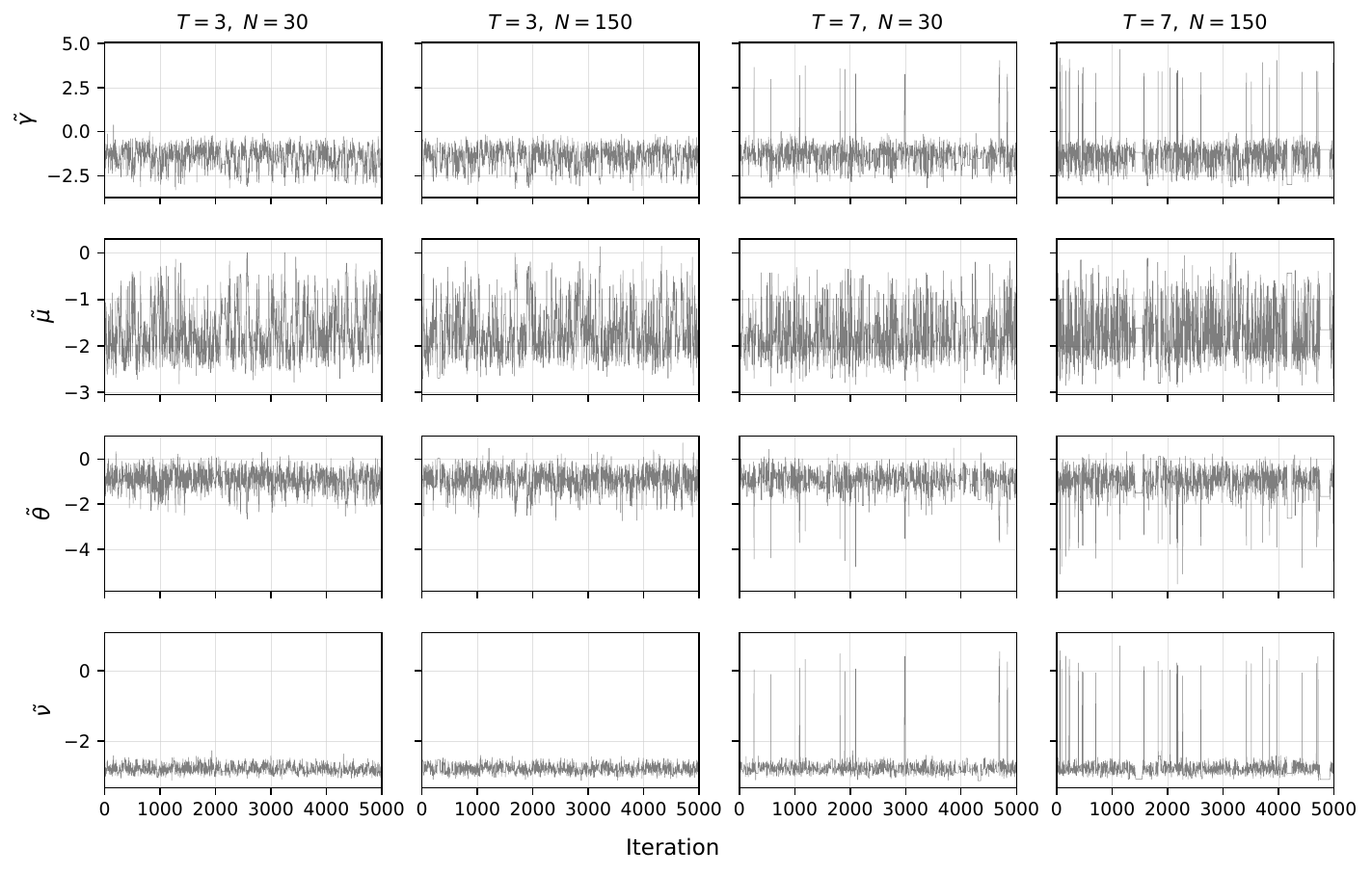}
    \caption{Trace plots of the transformed parameters under MAD-Path with varying horizon $T$ and number of steps $N$ on the fetal deaths in litters example.}
    \label{fig:fetal-deaths-traces}
\end{figure}

\paragraph{Sensor network localization (Section~\ref{sec:sensor:network})}
Figure~\ref{fig:sensor-network} shows the true location and posteriors of sensors. The wall-clock time for MAD-Path, including score training, is 14.83~s; the wall-clock times for the two APT-DEO chains are 14.59~s and 13.86~s, respectively.

\begin{figure}[htbp]
    \centering
    \includegraphics[width=\linewidth]{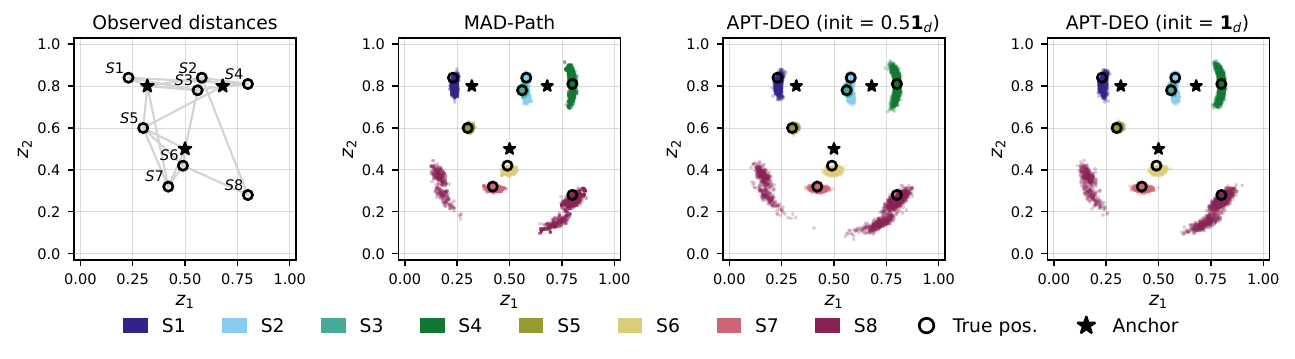}
    \caption{The left panel shows the true locations of the eleven sensors, with the three known locations marked as $\star$ and observed distances indicated by lines. The right three panels represent the posterior distributions of sensor locations obtained by MAD-Path and APT-DEO with two initializations after burn-in. }
    \label{fig:sensor-network}
\end{figure}

\paragraph{Seemingly unrelated regression (Section~\ref{sec:sur})}

The design vectors $X_m$ are drawn i.i.d.\ from $\mathcal{N}(0, I_N)$, the true coefficients $\boldsymbol{\beta}_m$ are drawn i.i.d.\ from $\mathcal{N}(0,1)$, and the cross-equation covariance is set to $\Sigma = (AA^\top)/M + I_M$ with $A \sim \mathcal{N}(0, I_{M\times M})$. 
The wall-clock time for MAD-Path, including score training, is 10.82~s; for APT-DEO, the cold-started chain takes 6.1~s, and the two warm-started chains (APT-DEO~1 and APT-DEO~2 in Figure~\ref{fig:sur-hist}) take 5.68~s and 5.87~s, respectively.
As in the sensor network example (Figure~\ref{fig:sensor-network-traces}), MAD-Path benefits from a natural warm start and so exhibits essentially no transient phase (Figure~\ref{fig:sur-trace}). APT-DEO with cold start has an average swap acceptance rate of 0.242.

\begin{figure}[htbp]
    \centering
    \includegraphics[width=\linewidth]{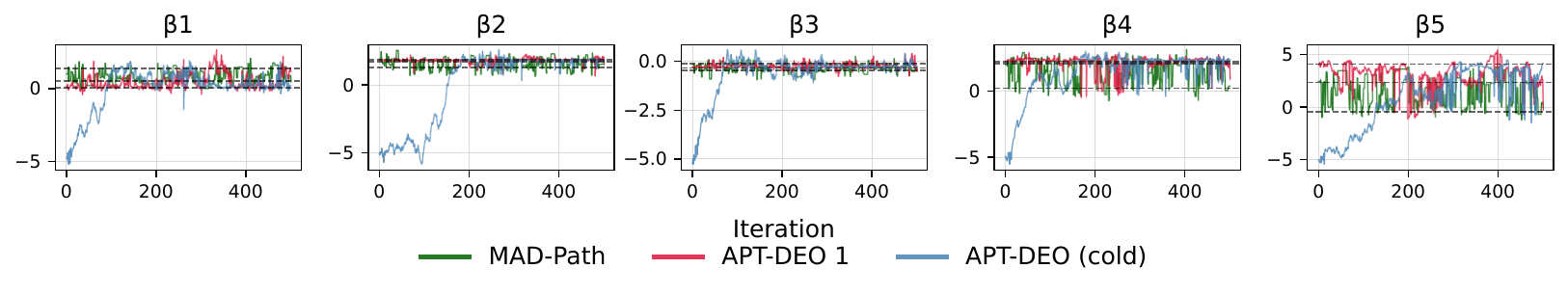}
    \caption{Trace plots for the first 500 iterations of MAD-Path and APT-DEO with varying initializations. Dashed lines indicate the oracle mode locations. APT-DEO 1: the same chain displayed in Figure~\ref{fig:sur-hist}; APT-DEO (cold): APT-DEO with a cold start. }
    \label{fig:sur-trace}
\end{figure} 

\bibliographystyle{apalike}
\bibliography{ref}

\end{document}